\documentclass[12pt, letterpaper, twoside]{article}
\usepackage[utf8]{inputenc}
\usepackage{amsmath}
\usepackage{array}
\usepackage{longtable}
\usepackage[square,sort,comma,numbers]{natbib}
\usepackage[figurename=Algorithm]{caption}
\usepackage{exptex}
\usepackage{expthmi}
\usepackage{setspace}
\usepackage{graphicx}
\usepackage{amsmath}
\usepackage{mathtools}
\usepackage{hyperref}
\usepackage{url}
\usepackage[title]{appendix}
\usepackage{mathtools}
\DeclarePairedDelimiter\ceil{\lceil}{\rceil}

\usepackage{simplemargins}
\setleftmargin{1in}
\setrightmargin{1in}
\settopmargin{1.9375in}
\setbottommargin{1.9375in}
\usepackage{xcolor}
\hypersetup{
    colorlinks,
    linkcolor={red!50!black},
    citecolor={blue!50!black},
    urlcolor={blue!80!black}
}
\usepackage{algorithm}
\usepackage{algorithmicx}
\usepackage{algpseudocode}

\def\cB{{\cal B}}

\def\cT{{\cal T}}

\newcommand{\HH}{{\bf H}}
\newcommand{\hh}{{\bf h}}

\newcommand{\ddd}{{\bf d}}
\newcommand{\D}{{\bf D}}

\newcommand{\A}{{\bf A}}

\newcommand{\es}{\mbox{E}(s^2)}
\newcommand{\smax}{s_{\max}}
\newcommand{\freq}{f_{\smax}}

\usepackage{amssymb}

 \usepackage{lineno}

\author{
	Luis B. Morales\\ Unidad Acad. IIMAS Estado de Yucat\'an, Universidad \\ Nacional Aut\'{o}noma de M\'exico, Yuc, M\'erida, Mexico \and	
	Dursun A. Bulutoglu\\ Air Force Institute of Technology Wright-Patterson Air\\ Force Base, Ohio USA 
}

\date{\today}

\title{A Bit-Parallel Tabu Search Algorithm for Finding E($s^2$)-Optimal and Minimax-Optimal
	Supersaturated Designs
}
\begin{document}
	\maketitle
\begin{abstract}
We prove the equivalence of two-symbol
supersaturated designs (SSDs)  with $N$ (even) rows, $m$ columns, 
$s_{\rm max} = 4t +i$, where $i\in\{0,2\}$, $t \in \mathbb{Z}^{\geq 0}$  and resolvable
incomplete block designs (RIBDs) 
whose any two blocks intersect in at most
$(N+4t+i)/4$ 
points.
Using this equivalence, we formulate the search for two-symbol 
$\es$-optimal and minimax-optimal SSDs with $s_{\max} \in \{2,4,6\}$ as a search for 
RIBDs whose blocks  intersect accordingly. 
This  allows developing
a bit-parallel tabu search (TS)  algorithm. 
The TS algorithm found $\es$-optimal 
and minimax-optimal SSDs achieving 
the sharpest known $\es$ lower bound 
with $s_{\max} \in \{2,4,6\}$ of sizes
$(N,m)=(16,25), (16,26), (16,27), (18,23),(18,24),(18,25),(18,26),(18,27),(18,
28),$ $(18,29),(20,21),(22,22),(22,23),(24,24)$, and $(24,25)$.
In each of these cases no such SSD could previously be found.
\end{abstract}
\section{Introduction}
Two symbol supersaturated designs (SSDs) are two symbol arrays in which
the number of  rows is less than or equal to the number of
columns. 
Throughout this paper, an SSD refers to a two-symbol SSD. 
An SSD
with $N$ rows  and $m$ columns is represented by an $N \times m$
matrix, and will be denoted by $\D(N,m)$ or simply by $\D$.  
Each entry of $\D$ is $\pm 1$, 
and the frequencies of $+1$ and  $-1$ are the same in each column. 
Moreover, $\D$ has no two columns   such that $\ddd_i =  \ddd_j$ or $\ddd_i = -\ddd_j$.

The $\es$  optimality 
criterion was defined  in~\cite{booth}, for comparing two-symbol SSDs.
The $\es$ criterion  compares two $\{-1,1\}$-arrays   of the same size by picking the one that minimizes 
\begin{equation*} \label{OF}
	\es = \sum_{i<j} \frac{ s_{ij}^2}{{m \choose 2}},
\end{equation*}
where $s_{ij}$ is the  $(i,j)$th entry of the matrix $\HH^{\top}\HH$ for a $\{-1,1\}$-array $\HH$.
The term $s_{ij}$ measures the degree of non-orthogonality
between the $i$th and $j$th columns. An SSD  is called {\em $\es$-optimal} if no SSD with the same number of rows and columns having a smaller $\es$ value exists.
For $\{-1,1\}$-arrays, let $s_{\rm max}={\rm max}_{i<j}|s_{ij}|$ and $f_{s_{\rm max}}$ be the frequency of $s_{\rm max}$ in $\{|s_{ij}|\}_{i<j}$.
The minimax criterion proposed in~\cite{booth}, minimizes $s_{\rm max}$ first, and  $f_{s_{\rm max}}$ second.  An SSD is called {\em minimax-optimal} if no other SSD with
the same size has a lower $s_{\rm max}$ or the same $s_{\rm max}$
with a smaller $f_{s_{\rm max}}$,
see~\cite{ryan}. 
In~\cite{ryan},    a search algorithm generalizing the exchange algorithm of~\cite{ngyuen} was provided 
to construct  $\es$-optimal and minimax-optimal SSDs.
An $\es$-optimal and minimax-optimal SSD with $16$ rows and $60$ columns was
found by~\cite{butler}. 
In~\cite{koukouvinosSA},  
a simulated annealing algorithm was used 
for finding $\es$-optimal and minimax-optimal
cyclic SSDs. In~\cite{koukouvinosSA2}, tabu search (TS) was used to construct $\es$-optimal SSDs with good properties by constructing supplementary difference sets.
In~\cite{gupta-morales}, a TS procedure for constructing
$\es$-optimal and minimax-optimal
$k$-circulant  SSDs was implemented.
Recently,  all isomorphism classes of $\es$-optimal and minimax-optimal  $k$-circulant SSDs with 
$N=6,  10,  14,  18,  22,  26$ rows,
$m=k(N-1)$ columns, and
$s_{\rm max} \in \{ 2,6\}$ were classified  in a computer search 
by~\cite{morales2014}.
They also classified all isomorphism classes of $\es$-optimal and minimax-optimal $k$-circulant SSDs  with
$N\equiv 0$ (mod 4) and   $s_{\rm max} =4$. For a comprehensive review of SSDs, see~\cite{georgiou}.

In Section~\ref{sec:es}, we provide some background material on $\es$ lower bounds and the  $\es$ and minimax optimality of SSDs. 
In Section~\ref{sec:equiv},   we  prove an equivalence between SSDs  with  $N$ (even) rows, 
$m$ columns, and $s_{\rm max} = 4t+i$, where $i\in\{0,2\}$, $t \in \mathbb{Z}^{\geq 0}$ and resolvable
incomplete block designs (RIBDs) such that any two distinct blocks intersect in at most
$(N+4t+i)/4$ points.
In Section~\ref{sec:opt}, using this equivalence, we formulate the problem
of constructing $\es$-optimal and minimax-optimal SSDs with  $s_{\rm max} \in \{ 2,6\}$ as a problem to
find RIBDs whose blocks intersect in at most $(N + 4t + i)/4$ points
for $(t,i)\in  \{(0,2),(1,0),(1,2)\}$.
We formulate the construction of such RIBDs as an optimization 
problem.
Unlike many optimization problems where a good
approximate solution is sufficient, in the construction of such resolvable designs (as in
the construction of other combinatorial designs), the main goal is 
to find an optimum solution.
For this purpose,  an algorithm based on  
TS~\cite{gloverTS} is developed  
in Section~\ref{sec:alg}.
Sets (blocks)
with elements from 
a set $V$ of cardinality $N$ are  stored as
bit strings where the number of bits is equal to $N$.
Each bit corresponds to exactly one element of $V$. 
Thus, a set is represented by a bit string 
in which the bits corresponding to the elements of that set are $1$ and all others bits are $0$.
(For example, in $V=\{0,1,2,3,4,5\}$,  bit string $BB=010011$ represents the set $B=\{1,4,5\}$.)  Our data structure for storing sets (blocks)
is the same as that in~\cite{segunda}.
This allows us to exploit  bit-parallelism as in~\cite{segunda} for computing  the intersections of the sets (blocks) by using
bitwise operations.
Thus,  all our computations are made using
bit-parallel Boolean instructions, which in praxis (on a {\tt x86-64  CPU}) implies that
$64$ bits of data are processed at once.
This improves the overall performance by a factor of $\min(N,64)=N$ (as $N$ is less than $64$ in all the cases we studied). In 
Section~\ref{sec:alg}, we also provide the computational complexity analysis of our algorithm.
The implementation details of our algorithm are discussed in 
Section~\ref{sec:impl}.  We end the paper with concluding remarks in  Section~\ref{sec:conclude}.

The bit-parallel  TS algorithm was able to construct fifteen previously unknown  $\es$-optimal and minimax-optimal SSDs 
of sizes
$(N,m)=(16,25), (16,26), (16,27),(18,23), (18,$ $24), (18,25),  (18,26), (18,27), (18,28), (18, 29), (20,21), (22,22), (22,23), (24,24)$ and $(24,25)$. All these SSDs, their $\es$ values,  $s_{\rm max}$s, and $f_{s_{\rm max}}$s,  as well as their Gram matrices are provided in Appendix~\ref{appendix}. 
The newly found $\es$-optimal and minimax-optimal SSDs could not have been found by the $NOA_p$ algorithms in~\cite{ryan} for $p=2,4,8$
despite running these algorithms for a very long time. So, the TS algorithm  in this paper outperforms the $NOA_p$ algorithms at least for the  SSD cases  searched in this paper.
The bit-parallel TS algorithm also found all  $\es$-optimal and  minimax-optimal SSDs obtained by 
the $NOA_p$ algorithms in~\cite{ryan}. 
We  made a comparison of TS with the $NOA_p$ algorithms. This is because in TS, $NOA_4$ and $NOA_8$ the objective function was chosen with respect to both the $\es$ and the  minimax criteria. 
In addition, $NOA_p$
algorithms were successful in locating all $\es$ optimal SSDs achieving the Ryan and Bulutoglu $\es$ lower bound~\cite{ryan} for $N\leq16$
except the $14$ row, $16$ column case, where the  Ryan and Bulutoglu $\es$ lower bound for this case was recently improved~\cite{morales19}.
No other previous algorithm had solved so many cases.

SSDs are used in computer
experiments, software testing, medical, industrial and engineering experiments, chromatography (separation science), 
as well as in biometric applications.
In~\cite{ngyuen}, an SSD with $N=28$ runs (rows) and $m=54$ factors (columns) for a crash test experiment on a planned new four-wheel drive range was discussed, where the objective was to find the best possible subset of safety features among the $54$ proposed.
In~\cite{ngyuencheng}, an SSD with $N=16$ runs (rows) and $m=18$ factors (columns) to fine-tune $16$ potential factors affecting the thermal performance of project homes was proposed, where  two additional columns were used as blocking factors. 
For designing a multistage axial compressor (turbine engine), the design engineer has selected $27$ potentially important
factors~\cite{montgomery}. In~\cite{montgomery},  SSDs with $m=27$ factors (columns) and $N=16$, $N=20$, and $N=12$ runs (rows) from~\cite{ngyuen}, \cite{ngyuen}, and~\cite{wu} respectively were compared by using data from a computer experiment for the design of a multistage axial compressor. Our newly found $N=16$ run (row), $m=27$ factor (column) $\es$-optimal, and minimax-optimal SSD could have also been used in this study. 
In~\cite{lcgc}, an SSD with $N=12$ runs (rows) and $m=24$ factors (columns) was used in composite sampling 
for monitoring pesticide residues in water.
Our newly found $N=18$ run (row), $m=24$ factor (column) $\es$-optimal, and minimax-optimal SSD could also have been used for the same purpose. 
In~\cite{WuBook}, an $N=14$ run (row), $m=23$ factor (column) SSD was proposed in place of a Plackett-Burman design that had been previously used in~\cite{Williams68} for developing an epoxide adhesive system for bonding a polyester cord. We propose our newly found $N=18$ run (row), $m=23$ factor (column) $\es$-optimal, and minimax-optimal SSD for the same purpose.
\section{Lower Bounds for $\es$ and $\es$ and Minimax Optimality of SSDs} \label{sec:es}
\label{optimality_SSD}
In~\cite{ngyuen}   and~\cite{tang}, it was independently shown that
\begin{equation} \label{bound}
	\es \ge \frac{N^2(m-N+1)}{(m-1)(N-1)}.
\end{equation}
Bound~(\ref{bound}) can be achieved only if $m=q(N-1)$
and $N\equiv 0$ (mod $4$), or if $m=2q(N-1)$ and $N\equiv 2$ (mod $4$)
for some positive integer $q$.
Bound~(\ref{bound}) was independently improved in~\cite{butler2001} 
and~\cite{bulutoglu2004}. These improved lower bounds  are equal when they both apply.

Let $\lfloor{x}\rfloor^+ =\mbox{ max} \{0, \lfloor{x}\rfloor\}$ and
$\lceil{x}\rceil^+ =\mbox{ max} \{0, \lceil{x}\rceil\}$, where  $\lfloor{.}\rfloor$ and
$\lceil{.}\rceil$ are the floor and ceiling functions.
The following theorem by~\cite{ryan} provides an improved version 
of the~\cite{bulutoglu2004}  lower bound.
\begin{theorem} \label{newbounds}
	(\cite{ryan}) Let $m$ be a positive integer such that  $m>N-1$. Then 
	there exists a unique nonnegative integer $q$ (which depends on $N$ and $m$) such that  
	$-2N +2 < m-q(N-1) < 2N-2$  and $(m+q) \equiv 2$ (mod 4). Define $g= (m+q)^2N-q^2N^2-mN^2$.
	{\scriptsize
		\begin{description}
			\item[{\rm (a)}] If $N\equiv 0$ (mod 4), then
			\begin{equation*} \label{n4}
				E(s^2) \ge \begin{cases} 
					\frac{g +2N^2 -4N}{m(m-1)} & {\rm if} \;   |m-q(N-1)| < N-1,\cr       
					\frac{g -2N^2 +4N + 4N|m-q(N-1)|}{m(m-1)} & {\rm if} \; N-1<  |m-q(N-1)|\le \frac{3}{2} N-2, \cr                
					\frac{g+4N^2-4N}{m(m-1)} &  {\rm if} \; |m-q(N-1)| > \frac{3}{2} N-2.
				\end{cases}
			\end{equation*} 
			\item[{\rm (b)}] If $N\equiv 2$ (mod 4), then
			\begin{equation*}
				E(s^2) \ge \frac{4m(m-1)  + 64\lceil{m(m-1)(h-4)/64}\rceil ^+}{m(m-1)},
			\end{equation*}
			where for even $q$
			\begin{equation*} \label{n2evenq}
				h = \begin{cases} 
					\frac{g +2N^2 -4N +8}{m(m-1)}  & {\rm if} \;   |m-q(N-1)| < N-1,\cr       
					\frac{g -2N^2 + 20N + (4N-8)|m-q(N-1)| - 24}{m(m-1)}
					& {\rm if} \; N-1<  |m-q(N-1)|\le \frac{3}{2} N-3, \cr                
					\frac{g+4N^2-4N}{m(m-1)}&  {\rm if} \; |m-q(N-1)| > \frac{3}{2} N-3,
				\end{cases}
			\end{equation*} 
			and for odd $q$
			\begin{equation*} \label{n2oddq}
				h = \begin{cases} 
					\frac{g +2N^2 -4N}{m(m-1)} & {\rm if} \;   |m-q(N-1)| < N-1,\cr       
					\frac{g -2N^2 +4N + 4N|m-q(N-1)|}{m(m-1)} & {\rm if} \; N-1<  |m-q(N-1)|\le \frac{3}{2} N-1, \cr                
					\frac{g +4N^2 -12N + 8|m-q(N-1)|  + 8}{m(m-1)} &  {\rm if} \; |m-q(N-1)| > \frac{3}{2} N-1.
				\end{cases}
			\end{equation*} 
	\end{description}}
\end{theorem}
In this paper, we search for $\es$-optimal SSDs achieving the lower bound in Theorem~\ref{newbounds}  with $s_{\max}\leq 6$. 
The following theorem shows that  $\es$-optimality is a
sufficient condition for minimax optimality if 
$s_{\rm max} \leq 6$.
\begin{theorem} \label{minimax}
	(\cite{ryan}) Let $\D(N,m)$ be an $\es$-optimal SSD.
	\begin{description}
		\item[{\rm (a)}] If $N\equiv 0$ (mod 4) and $s_{\rm max}=4$, then $\D(N,m)$
		is minimax-optimal.  
		\item[{\rm (b)}] If $N\equiv 2$ (mod 4) and $s_{\rm max} \in \{2,6\}$, then
		$\D(N,m)$ is minimax-optimal.
	\end{description}
\end{theorem}
\section{The Equivalence Between  SSDs and RIBDs} \label{sec:equiv}
\label{SSDequivalenceRBIBD}
An incomplete block design (IBD)  with parameters $(v, b, r, h)$,
denoted by IBD$(v, b, r, h)$,
is a pair $(V, \cB)$ where  $V$ is a $v$-set of points and $\cB$ is a collection of
$b$ $h$-subsets ({\em  blocks}) of $V$, $h<v$. The   parameters must satisfy the condition
\begin{equation*}
	vr = bh, 
\end{equation*}
An IBD $(V,\cB)$ with parameters $(v, b, r, h)$  is called a {\em resolvable incomplete block design},
denoted by  RIBD, if  the collection
$\cB$ of blocks can be partitioned into $r$ subsets called {\em parallel classes}
of size $q = b/r$,  each of which partitions the point set.  

Henceforth, $N$ will denote a positive even integer greater 
than or equal to $8$.
Let $Q=( B_{1},B_{2})$ and  $Q^\prime =( B^\prime_{1},B^\prime_{2})$  
be two different parallel classes on $N$ points.
Define their {\em parallel class intersection matrix} (PCIM) as the
$2\times 2$ matrix $\A(Q,Q^\prime)$ with entries defined by
$a_{ij} = |B_{i}  \cap B^\prime_{j}|$, see~\cite{morales}.
An IBD$(N,2m,m,N/2)$ has $m$ parallel classes. 
Thus, for an arbitrary fixed  parallel class    $Q$ there are 
$m-1$  PCIMs  of the form $\A(Q,Q^\prime)$.
Since each point belongs to exactly one block of a 
parallel class, both the column and row sums in a PCIM are $N/2$. Then, by relabeling the 
blocks in parallel classes if necessary, each of the
$m-1$ PCIMs associated with $Q$ can be assumed to be one of  
\begin{equation}\label{PCIMs}
	{\mathbf \cT}_{i,t}=
	\left[
	\begin{array}{c c}
		\frac{N-i}{4}-t&\frac{N+i}{4} +t \\
		\frac{N+i}{4}+t& \frac{N-i}{4}-t \\
	\end{array}
	\right], \; t = 0,\ldots,  \frac{N-i}{4}, \;\mbox{for} \; N\equiv i \,\, \, \mbox{(mod 4)} \;  
\end{equation}
for $i\in\{0,2\}$.
For any $2 \times 2$ matrix $\A=[a_{ij}]$,  let 
$S(\A)=  |a_{11}-a_{12}-a_{21}+a_{22}|$, hence $S({\mathbf \cT}_{i,t}) = 4t + i$.
\begin{theorem} \label{SSDvsRIBD} 
	An  RIBD with
	parameters $(N, 2m,m,  N/2)$ 
	such that for any two distinct parallel classes $Q$ and $Q'$,
	$S(\A(Q,Q^\prime)) \leq 4t+i$ 
	exists
	if and only if a $N$ row, $m$ column, $\{-1,1\}$-array $\HH(N,m)$ with each column orthogonal to the all $1$s column and
	$s_{\rm max} = 4t +i$, where $i\in\{0,2\}$ exists. 
\end{theorem}
\begin{proof}
	Suppose that 
	$(B_{1,1}, B_{1,2}),\ldots, (B_{m,1}, B_{m,2})$ 
	are the parallel classes of the RIBD$(N, 2m,$ $m,N/2)$.
	Then for each parallel class  $(B_{\ell,1}, B_{\ell,2})$ ($1\leq \ell \leq m$), we define
	the column vector  $\hh_{\ell}$  as follows: 
	\begin{equation} \label{PC_Column}
		h_{\ell,p}= \begin{cases} \phantom{-}1& {\rm if} \;  p \in B_{\ell,1},  \cr                      
			-1&  {\rm if} \;   p \in B_{\ell,2}, 
		\end{cases}
	\end{equation}
	for each $1\leq p \leq N$, where 
	$h_{\ell,p}$ is the $p$th entry of  $\hh_{\ell}$. 
	Note that the frequencies of $+1$ and $-1$ are the same in each column constructed 
	from the RIBD.
	Hence, these $m$  columns form a $\{-1,1\}$-array $\HH(N,m)$ with each column orthogonal to the all $1$s column.
	For any two columns  $\hh_\ell$ and $\hh_j$ of $\HH(N,m)$ defined by  the parallel classes $Q_\ell$ and $Q_j$, we have
	$|s_{\ell,j}| =  S(\A(Q_\ell,Q_j)) \le 4t +i$ by~(\ref{PCIMs}) and~(\ref{PC_Column}). This implies that $s_{\rm max} = 4t +i$, where $i\in\{0,2\}$. 
	
	Conversely, suppose that  $\HH(N,m)$ is a $\{-1,1\}$-array with each column orthogonal to the all $1$s column.
	For each column $\hh_{\ell}$ ($1\leq \ell \leq m$)
	of the array $\HH(N,m)$, there exist two blocks $B_{\ell,1}$ and $B_{\ell,2}$ that 
	partition $V= \{1,2,\ldots, N \}$, such that if the $p$th entry of $\hh_{\ell}$ is $1$,  then
	$p$  is contained in the block $B_{\ell,2}$, otherwise $p$ is contained in $B_{\ell,2}$.   
	Clearly, these two blocks form a parallel class. Since the frequencies of $+1$ and $-1$ 
	in each column are both $N/2$, each  block has size $N/2$. Now,  $s_{\rm max} = 4t +i$, where $i\in\{0,2\}$ 
	implies that $S(\A(Q_\ell,Q_j)) \le 4t +i$ for any two   parallel classes $Q_\ell$ and $Q_j$ defined by the $\ell$th and $j$th columns. 
\end{proof}
Next, we provide an example for Theorem~\ref{SSDequivalenceRBIBD}, where 
$(B_{1,1},B_{1,2}),\ldots,(B_{8,1},B_{8,2})$ is an RIBD$(6,16,8,3)$  corresponding to a $6$ row, $8$ column, 
$\{-1,1\}$-array  with each column orthogonal to the all $1$s column and
$\smax = 6$.
\begin{example}
	$$\HH(6,8)=\left[\begin{array}{rrrrrrrr}
		-1& 1&-1&-1&-1& 1& 1& 1 \\ 
		1&-1&-1& 1& 1& 1&-1& 1 \\
		-1&-1& 1&-1&-1&-1&-1& 1\\ 
		-1& 1&-1& 1& 1& 1&-1&-1\\ 
		1& 1& 1& 1&-1&-1& 1&-1\\ 
		1& -1& 1&-1& 1&-1& 1&-1
	\end{array}\right]
	$$
	$B_{1,1}=\{2,5,6\}$, $B_{1,2}=\{1,3,4\}$,  $B_{2,1}=\{1,4,5\}$, 
	$B_{2,2}=\{2,3,6\}$,  $B_{3,1}=\{3,5,6\}$, $B_{3,2}=\{1,2,4\}$,
	$B_{4,1}=\{2,4,5\}$, $B_{4,2}=\{1,3,6\}$,  $B_{5,1}=\{2,4,6\}$, 
	$B_{5,2}=\{1,3,5\}$,  $B_{6,1}=\{1,2,4\}$, $B_{6,2}=\{3,5,6\}$, 
	$B_{7,1}=\{1,5,6\}$, $B_{7,2}=\{2,3,4\}$,  $B_{8,1}=\{1,2,3\}$, $B_{8,2}=\{4,5,6\}$.
\end{example}
\section{The Optimization Problem} \label{sec:opt}
\label{optimization}
In this section, using the equivalence given in Theorem~\ref{SSDvsRIBD}, we formulate the problem of constructing an
$\es$-optimal and  minimax-optimal SSD with  $N$ (even) rows, $m$ columns and $s_{\rm max} = 4t +i$ for $(t,i)\in  \{(0,2),(1,0),(1, 2)\}$ (i.e., $s_{\max} \in \{2,4,6\}$) as a discrete optimization problem.
The following theorem is used  to define the objective function for this optimization problem.

\begin{theorem} \label{propertiesRIBDs} 
	Let $\D$ be an SSD with $N$ rows and  $m$ columns, and 
	$(B_{1,1},B_{1,2}), \ldots,(B_{m,1},B_{m,2})$ be the parallel classes of the RIBD  defined by the columns of $\D$ according to Theorem~\ref{SSDvsRIBD}. Let $1\le h,p \le2$, and $\ell \not= j$, then the following hold.
	\begin{description}
		\item[{\rm (a)}] $|s_{\ell,j}| = |4|B_{\ell,p} \cap B_{j,h}| - N|$.
		\item[{\rm (b)}] $\es = \frac{1}{{m \choose 2}}\sum_{\ell<j} (4|B_{\ell,h} \cap B_{j,p}| - N)^2$.
		\item[{\rm (c)}]  $s_{\rm max} = 4t +i \Longleftrightarrow
		\frac{N-i}{4}-t \le |B_{\ell,h} \cap B_{j,p}|  \le \frac{N+i}{4}+t$,  where $i\in\{0,2\}$.
	\end{description}
\end{theorem} 
\begin{proof}
	By the proof of Theorem~\ref{SSDvsRIBD} and~(\ref{PCIMs}), we have  $|B_{\ell,p} \cap B_{j,h}|  =\frac{N-i}{4}-t$
	(or $\frac{N+i}{4} +t$) for some $1\le t \le \frac{N-i}{4}$. Then $|4|B_{\ell,p} \cap B_{j,h}| - N| = |\pm(i+4t)| = 4t+i=|s_{\ell,j}|$.
	Statements (b) and (c) follow from (a). 
\end{proof}
By Theorem~\ref{SSDvsRIBD},  a  feasible solution to our optimization problem is a RIBD with parameters $(N,2m,m,N/2)$.
However, since each parallel class of the RIBD
is uniquely determined by one of its blocks, a feasible solution reduces  to a set $\cB$  of $m$ blocks, $B_1, \ldots, B_m$ each of size $N/2$.  
Then based on Theorem~\ref{propertiesRIBDs}, we define the objective function as
\begin{equation*} \label{ObF0}
	g(\cB) = \sum_{\ell<j} \frac{(4|B_{\ell} \cap B_{j}| - N)^2}{{m \choose 2} }.
\end{equation*} 
Computing the intersections of blocks is the  bottleneck in computing $g(\cB)$. 
We used the bit array data structure to store each set (block) 
as in~\cite{segunda}. This allowed speeding up the calculation of intersections of blocks by using  bit-parallel Boolean instructions.
To determine the number of points of intersection of two blocks, we used  the SSE4.2 SIMD instruction,
$\_$mm$\_$popcnt$\_$u64, 
included in the recent general-purpose processors.  It counts the number of bits set to $1$ in a word of  $64$ bits.
Thus, in the {\tt C} language, $|B_{\ell} \cap B_{j}|$ is calculated by  
$\_$mm$\_$popcnt$\_$u64($BB_\ell \& BB_j$), where $BB_h$ is the 
binary representation of the block $B_h$ with $N \leq 64$.
These instructions  increase the speed by a factor of $\min (N,64)=N$.

To construct $\es$-optimal and minimax-optimal SSDs  based on 
Theorem~\ref{minimax}, it is necessary
to require that    $s_{\rm max} = 4t +i$ for $(t,i)\in  \{(0,2),(1,0),(1,2)\}$ (i.e., $s_{\max} \in \{2,4,6\}$). For this purpose, we define 
\begin{equation} \label{weigth}
	w(\ell,j)= \begin{cases} 1& {\rm if} \; \frac{N-i}{4}-1 \le |B_{\ell} \cap B_{j}|  \le \frac{N+i}{4}+1,  \cr                      
		b(N,m)&     \mbox{otherwise}, 
	\end{cases}
\end{equation}  for $1\le \ell<j \le m$, 
where $b(N,m)$ is the lower bound given in  Theorem~\ref{newbounds}. Then, we modify the objective function to:
\begin{equation} \label{ObF}
	f(\cB) = \sum_{\ell<j} w(\ell,j)\frac{(4|B_{\ell} \cap B_{j}| - N)^2}
	{ {m \choose 2} }. 
\end{equation} 
It follows from Theorem~\ref{propertiesRIBDs}~(c) and~(\ref{weigth}) that if the objective function~(\ref{ObF})  reaches the value $b(N,m)$,
then we have  $s_{\rm max} = 4t +i$ for $(t,i)\in  \{(0,2),(1,0),(1,2)\}$.  Hence,
Theorems~\ref{minimax},~\ref{SSDvsRIBD},  and~\ref{propertiesRIBDs}~(a)-(b) imply that an
$\es$-optimal and  minimax-optimal SSD with  $N$ (even) rows, $m$ columns and $s_{\rm max} = 4t +i$ for $(t,i)\in  \{(0,2),(1,0),(1,2)\}$
is found whenever the  objective function~(\ref{ObF})  achieves $b(N,m)$ for a RIBD.
\section{Tabu Search for SSDs} \label{sec:alg}
\label{TS}
The TS algorithm introduced 
by~\cite{gloverTS}  is an iterative metaheuristic technique used 
to search for  a solution that minimizes an objective function $f$  
over a set of feasible solutions $X$. TS has been used successfully to  construct $D$-optimal designs, constant weight codes, 
$1$-rotational resolvable balanced incomplete block designs, 
covering designs and  $\es$-optimal 
and minimax-optimal $k$-circulant SSDs, 
see~\cite{dai2009,gupta-morales,joo-bong,morales2001, nurmela1997}. 

TS is based on a neighborhood search (NS). 
In NS, each feasible solution $x$ has an 
associated set of neighbors, $N(x)\subset X$, called the {\it 
	neighborhood} of $x$.
It starts with a given initial feasible solution and searches 
the set $X$ by moving from one solution to another  in 
its neighborhood.
At each iteration, 
a move from the current solution $x$ to a best one  $x^\prime$ in $N(x)$ regardless of whether 
$f(x^\prime) \leq f(x)$ is made.  If more than one solution has the same minimum value, the tie is broken randomly.
However, the main shortcoming of NS is cycling through a set of solutions, i.e., keeping on revisiting the same set of solutions.
To prevent cycling, TS     maintains a list called the {\it tabu list}
$T$ of {\it length} $|T|=M$. 
Each move in $T$ is removed  after $M$ iterations.

Sometimes, the tabu list may forbid certain desirable
moves, such as those that lead to a better solution than
the best one found so far. An {\it aspiration criterion} $aspF$ is 
introduced to cancel the tabu status of a move when this move is 
judged useful.

TS stops when the objective function  reaches the lower bound $b(N,m)$.
However, there is no guarantee of
reaching the lower bound, and the search process is stopped if the number of
iterations used without improving the best solution exceeds a preset $nitmax$ limit. 
\begin{figure}
		\begin{center}
	\begin{tabbing}
		00\=000\=000\=000\=000\=000\=000\=\hskip 9.5em\=00=\kill
		1\> Input $N$, $m$, $nitmax$, $b(N, m)$, $M$. \\
		2\> {\bf Generate} an initial RIBD$(N, 2m, m, N/2)$ (solution) $\cB_0$ randomly;  \\
		3\> {\bf Set} $\cB_{best} := \cB_0$, $T :=\emptyset$, $r := rbest := 1$, $fbest :=aspF:= f(\cB_{best})$;  \\
		4\> {\bf while} ($r - rbest \leq nitmax$ \& $fbest > b(N, m)$) {\bf do} \\
		5\>\> {\bf Set} $min = \infty$;  \\
		6\>\> {\bf for} $\cB' \in N(\cB_0)$ {\bf do} \\
		7\>\>\> {\bf Set} $s:=$ move from $\cB_0$  to $\cB'$; \\
		8\>\>\> {\bf if} ($f(\cB') \leq min$ \&  ($s \notin T$ or $f(\cB') < aspF$) ) {\bf then} \\
		9\>\>\>\> {\bf if} ($f(\cB') == min$) {\bf then} \\
		10\>\>\>\>\> {\bf Set} $\cB'' := \cB'$ with 50\% probability;  \\
		11\>\>\>\> {\bf else if} \\
		12\>\>\>\>\> {\bf Set} $\cB'' := \cB'$, $min =f(\cB')$ ;  \\
		13\>\>\>\> {\bf end if} \\
		14\>\>\> {\bf end if} \\
		15\>\> {\bf end  for} \\
		16\>\> {\bf if} ($min < fbest$) {\bf then}  \\
		17\>\>\> {\bf Update} $\cB_{best} := \cB''$; \\
		18\>\>\> {\bf Set} $fbest := aspF := min$,\, $rbest :=r$;  \\
		19\>\> {\bf end if} \\
		20\>\> {\bf Update} $T := T \cup \{ \mbox{move from } \cB'' \mbox{ to } \cB_0 \}$; \\
		21\>\>{\bf Update} $\cB_0$ := $\cB''$;  \\
		22\>\> {\bf if} ($|T| > M$) {\bf then} \\
		23\>\>\> remove oldest move from $T$ \\
		24\>\>{\bf end if} \\
		25\>\> {\bf Set} $r := r + 1$;  \\
		26\>\ {\bf end while} \\
		27\>\ {\bf Output} $\cB_{best}$, $fbest$. \\
	\end{tabbing}
\caption{The TS algorithm} \label{codeTS}
\end{center}
\end{figure}
Two IBDs with parameters $(N,2m,m,N/2)$
are defined as {\em neighbors} if they are identical for every parallel class but one, and in that parallel class there are exactly two points that switch blocks. A swap move is entirely  determined by the vector $(\ell,u,w)$, where points $u$ and $w$ are switched in the parallel class $\ell$.
The definitions of the neighborhood and the objective function  in
 Section~\ref{sec:opt}  allow  calculating the change in the
objective function~(\ref{ObF}) value without recomputing the objective function~(\ref{ObF}). 
The only blocks that change 
after the move ($\ell,u,w)$ are $B_{\ell,1}$ and $B_{\ell,2}$.

Whenever  the 
points $u$ and $w$  switch blocks in the parallel class $\ell$, the tabu list forbids the exchange of the points $u$ and $w$ at the parallel class $\ell$ in the subsequent  $M$  iterations.  
Formally, the tabu list consists of vectors $(\ell,u,w)$, where the points $u$ and $w$ were  forbidden to be  exchanged 
during the preceding  $M$  iterations,  in the parallel class $\ell$. The tabu list length $M$ was adjusted experimentally. 
For the problem instances in this paper, the best $M$ 
seems to be some integer
between $6$ and $8$. The pseudocode of our TS algorithm is 
presented in Algorithm~\ref{codeTS}.

In the above described algorithm, the computation time is mainly spent on iterations. Hence, we next provide the complexity analysis of each iteration.
Let $\mathcal{B} = \{B_1 , \ldots , B_{\ell} , B_{\ell+1},$ $\ldots,  B_m \}$ be an RIBD (a feasible solution to our optimization problem). 
In Algorithm~\ref{codeTS}, for each block $B_{\ell}$ that changes to
 $B_{\ell}'$ 
 there are $m-1$ blocks,
 $\{B_1 , \ldots , B_{\ell-1},$ $ B_{\ell+1},\ldots, B_m \}$ that do not change. Let $\mathcal{B}' = \mathcal{B} - 
 \{B_{\ell} \} \cup \{B'_{\ell}
 \}$. Then the
 objective function~(\ref{ObF}) is updated according to
$$f(\mathcal{B}' ) = f(\mathcal{B}) + \frac{1}{{m \choose 2}}\sum_{j=1, j \neq \ell}^m 	w(\ell,j) \left[(4|B'_{\ell}\cap B_{j}|-N)^2- (4|B_{\ell}\cap B_{j}|-N)^2\right].$$
Since the intersection of two blocks is performed in
$\ceil*{N/64}$ bitwise operations, the complexity to update the objective function after a move is $\mathcal{O}(mN)$.
The only  blocks that change after the move $(\ell, u,w)$  are 
$B_{\ell,1}$ and $B_{\ell,2}$ ($1\leq \ell \leq m$) and there are $ (N/2)^2 $ possible changes to $(B_{\ell,1},B_{\ell,2})$. Hence, the size of the neighborhood of any
RIBD is  $ m(N/2)^2$. Then, the
overall time spent for each iteration of this algorithm is 
$$
m\times\frac{N^2}{4}\times(m-1) \times   \ceil*{\frac{N}{64}}.
$$

Let $I(m,N)$ denote the expected number of  iterations of the algorithm for the $m$ column and $N$ row case.  Then, for $N\leq 64$, 
the expected time complexity for each run of this algorithm is $I(m,N)\times\mathcal{O}(m^2 N^2)$.
For the most difficult cases the algorithm was run for $4{,}000{,}000$  times. So, the overall expected running
time of the algorithm was  $4{,}000{,}000 \times I(m,N) \times \mathcal{O}(m^2 N^2)$. If we had not used bit-parallelism, then the expected time 
complexity for each run would have been $I(m,N)\times\mathcal{O}(m^2N^3)$. 

\section{Implementation Details} \label{sec:impl}  
The TS algorithm described above was programmed in  {\tt C}
and all computations were carried out on a 2.67 GHz or 2.4 GHz processor.
The source code of the algorithm 
can be requested by sending an email to the first author.

The TS algorithm was used  to construct  fifteen $\es$-optimal and minimax-optimal 
SSDs achieving the $\es$ lower bound of~\cite{ryan} with $s_{\max}\leq6$ 
for $N=16,18,20,22,24$  and $N \le m \le 30$.   
The existence 
question for each of these SSDs was previously unknown. 

Initial computational experiments show that the strategy of running the algorithm a larger number of times 
with a smaller $nitmax$
  is better than running the algorithm a smaller number of times with a larger  $nitmax$. 
For example, for $N = 18$ and $m=23$, 
$8{,}000$ runs of the algorithm with $nitmax=200$ found $3$ optimum solutions, whereas only $1$ optimum solution 
was found by $2{,}000$ runs with $nitmax=800$. 
Then, the TS procedure was carried out at most $80{,}000$ times   
using $nitmax = 300$ on each instance tested.
With these values, the TS algorithm did not produce any optimum solutions for  $(N,m)= (16,27), (18,24), (18,29), (24,25)$. However, the best SSDs found for these cases had $s_{\rm max} = 4t+i$ for $(t,i)\in  \{(1,0),(1,2)\}$, and an $\es$ value equal to   
$b(N,m) + 32/(m(m-1))$ and  $b(N,m)+64/(m(m-1))$,  
where $f_{s_{\rm max}}$  is larger by one  than the  $f_{s_{\rm max}}$ of a hypothetical $\es$-optimal and minimax-optimal SSD achieving the $\es$ lower bound of~\cite{ryan} with $s_{\rm max}\leq 6$.
Then, for these cases, the TS algorithm was carried out at most $4{,}000{,}000$ times with  $nitmax = 450$.
Since the TS algorithm runs are completely independent and no information is exchanged, an  independent-thread parallelization strategy was  used.  
Different random number seed values were used
to avoid an overlapping search. 

The bit-parallel  TS was able to construct fifteen previously 
unknown $\es$-optimal and  minimax-optimal SSDs.
Table~1 lists all cases where the best obtained  $\es$-optimal SSD is minimax-optimal.
Each row of this table corresponds to an SSD. The first column shows the number of rows and columns for the SSDs. Column nRUNs  gives the number of runs it took the TS algorithm to find   an    $\es$-optimal  and minimax-optimal SSD. 
The last column gives the  CPU time.
The $N=16$ row, $m=27$ column case was solved in $585.4$ CPU hours.
However,
with the independent-thread parallelization approach, this  search took only $16.1$ hours, 
using a cluster with 40 threads.   
\begin{table}[ht]\label{table1}
		\begin{center}
	\begin{tabular}{r r@{\hskip 1.3in} r@{\hskip 1.3in} r }
		\hline
		$N$ & $m$ & nRUNs & CPU time \\
		\hline
		16 & 25  &      95,989             & 211.5 minutes \\
		& 26  &    190,049            & 1.1 days \\ 
		& 27  & 3,954,798                     & 24.4 days \\
		18  & 23  &         1,232                & 3.5 minutes   \\
		& 24  &     880,559                    & 2.7 days \\
		& 25  &      40,058               & 241.6 minutes \\
		& 26  &        12,516               & 74.2 minutes  \\
		& 27  &        26,689               & 174  minutes\\
		& 28  &        14,243              & 96.9 minutes \\
		& 29  &       392,285          &  3.1 days \\ 
		20  & 21  &             19                 & 0.1   minutes \\
		22  & 22  &               6                  & 0.05 minutes \\
		& 23  &                7                  & 1.0  minutes \\
		24  & 24  &           509                & 4.0 minutes  \\
		& 25  &      112,976              &  1.6 days  \\ 
		\hline
	\end{tabular}
	\caption{$\es$-optimal and minimax-optimal SSDs obtained by the TS algorithm}
	\end{center}
\end{table}

In Table 1, we do not observe that	the CPU time always increases with the number of columns or rows. 
There are also big increases on	the	CPU	times,	just by	the	 addition of one column.
There are two reasons for these observations.
Firstly, the geometry of the problem can change as the number of columns
of the sought after SSD increases. 
In particular, let $\rho(N,m)$ and $\rho'(N,m)$
be the ratio of the number of all $N$ row, $m$ column, 
$\es$-optimal and minimax-optimal SSDs to the number 
of all and all locally optimum $N$ row, $m$ column SSDs.
It is possible that for fixed $N$, $\rho(N,m)$ and/or $\rho'(N,m)$ is not a non-increasing function of $m$. This is mainly because not every $\es$-optimal and minimax-optimal, $m+1$ column, $N$ row SSD can be obtained by adding a column to an $\es$-optimal and minimax-optimal, $m$ column, $N$ row SSD.
Then the probability of finding an $\es$-optimal and minimax-optimal SSD in one iteration of the TS algorithm may 
actually increase going from $m$ columns to $m+1$ columns.
Secondly, TS is not a deterministic algorithm and the
CPU times are random with potentially large variances. Large variances may easily blur an increasing pattern. In fact, this is more of an issue 
for the previously unsolved difficult cases.

The only other algorithm that is competitive with the TS algorithm is the $NOA_p$ algorithm. 
For the $NOA_p$ algorithm, each new random  starting SSD  is independently picked from the previous random starting SSDs. So, each trial of the $NOA_p$ algorithm with a new random starting SSD can be thought as a Bernoulli trial with a success probability of $p$ of finding an SSD that achieves the best known $\es$  and minimax lower bounds. Then we can use 
\[Y = \begin{array}{l}\text{The number of trials before 
		finding an SSD that achieves}\\  
	\text{the  $\es$ and minimax lower bounds in~\cite{ryan}}
\end{array}\]
as a surrogate for the time it takes to find an SSD that achieves the best known $\es$ and minimax lower bounds. Now, $Y \sim \text{Geometric}(p)$ where 
$$\text{E}(Y)=\frac{1-p}{p},\quad \text{Var}(Y)=\frac{1-p}{p^2}.$$
For  the most difficult cases, $p$ is very small,  making $\text{Var}(Y)$ very large.
So, it is possible to get very lucky and find a solution very quickly or be very unlucky and not be able to find a solution after a very long time.
Hence, for the most difficult cases, we do not gain as much information by comparing CPU times as in the case of exact algorithms. 
When TS is used, we do not know   the distribution of the random variable $Y$  as  the independence of trials is no longer a valid assumption.
Finding an empirical distribution for $Y$ would require repeating the 
computational experiments in the paper a large number of times.
This is neither feasible due to resource and time constraints, nor worthed as all cases solved by the TS algorithm have no corresponding CPU times based on the $NOA_p$ algorithm (or any other algorithm).  No corresponding CPU times exist because  the $NOA_p$ algorithm failed to solve them despite being run for a very long time.
\section{Concluding Remarks}\label{sec:conclude}
In this paper, we developed a heuristic  algorithm for finding E($s^2$)-optimal and minimax-optimal SSDs that is more effective than the previously known most effective algorithm for the same purpose. Our algorithm  brings 
fifteen cases of E($s^2$)-optimal and minimax-optimal SSDs within computational reach by taking advantage of the equivalence between an SSD and an RIBD described in the proof of Theorem~\ref{SSDvsRIBD}  and bit-parallelism from~\cite{segunda}.
\section*{Acknowledgements}
The views expressed in this article are
those of the authors and do not reflect the official policy or
position of the United States Air Force, Department of Defense, or
the US Government. The authors 
thank the High Performance Computing Lab at
IIMAS-UNAM for providing computing resources. 
This paper is also published on Arxiv~\cite{arxiv2023}.
\bibliographystyle{plain}
\bibliography{BitparallelBib}

\begin{thebibliography}{10}

\bibitem{booth}
K.~H.~V. Booth and D.~R. Cox.
\newblock Some systematic supersaturated designs.
\newblock {\em Technometrics}, \textbf{4}:489--495, 1962.

\bibitem{bulutoglu2004}
D.~A. Bulutoglu and C.~S. Cheng.
\newblock Construction of $\es$-optimal supersaturated designs.
\newblock {\em Ann. Stat.}, \textbf{32}:1662--1678, 2004.

\bibitem{butler}
N.~A. Butler.
\newblock Minimax $16$-run supersaturated designs.
\newblock {\em Stat. Probab. Lett.}, \textbf{73}:139--145, 2005.

\bibitem{butler2001}
N.~A. Butler, R.~Mead, K.~M. Eskridge, and S.~G. Gilmour.
\newblock A general method of constructing $\es$-optimal supersaturated
  designs.
\newblock {\em J. R. Stat. Soc. Series B Stat. Methodol.},
  \textbf{63}:621--632, 2001.

\bibitem{dai2009}
C.~Dai, B.~Li, and M.~Toulouse.
\newblock A multilevel cooperative tabu search algorithm for the covering
  design problem.
\newblock {\em J. Combin. Math. Combin. Comput.}, \textbf{68}:33--65, 2009.

\bibitem{georgiou}
S.~Georgiou.
\newblock Supersaturated designs: A review of their construction and analysis.
\newblock {\em J. Stat. Plan. Infer.}, \textbf{144}:92--109, 2014.

\bibitem{gloverTS}
F.~Glover.
\newblock Tabu search {I}.
\newblock {\em ORSA J.~Comput.}, \textbf{3}:190--206, 1989.

\bibitem{gupta-morales}
S.~Gupta and L.B. Morales.
\newblock Constructing $\es$-optimal and minimax-optimal $k$-circulant
  supersaturated designs via multi-objective tabu search.
\newblock {\em J. Stat. Plan. Infer.}, \textbf{142}:1415--1420, 2012.

\bibitem{montgomery}
D.~R. Holcomb, D.~C. Montgomery, and W.~M. Carlyle.
\newblock The use of supersaturated experiments in turbine engine development.
\newblock {\em Quality Engineering}, \textbf{19}:17--27, 2007.

\bibitem{joo-bong}
S.~J. Joo and J.~Y. Bong.
\newblock Construction of exact ${D}$-optimal designs by tabu search.
\newblock {\em Comput. Stat. Data Anal.}, \textbf{21}:181--191, 1996.

\bibitem{koukouvinosSA}
C.~Koukouvinos, K.~Mylona, and D.~E. Simos.
\newblock $\es$-optimal and minimax-optimal cyclic supersaturated designs via
  multi-objective simulated annealing.
\newblock {\em J. Stat. Plan. Infer.}, \textbf{138}:1639--1646, 2008.

\bibitem{koukouvinosSA2}
C.~Koukouvinos, K.~Mylona, and D.~E. Simos.
\newblock An algorithmic construction of $\es$-optimal supersaturated designs.
\newblock {\em J. Stat. Theory Pract.}, \textbf{5}:357--367, 2011.

\bibitem{morales2001}
L.~B. Morales.
\newblock Two new $1$-rotational $(36,9,8)$ and $(40,10,9)$ {RBIBD}s.
\newblock {\em J. Combin. Math. Combin. Comput.}, \textbf{36}:119--126, 2001.

\bibitem{arxiv2023}
L.~B. Morales and D.~A. Bulutoglu.
\newblock A bit-parallel tabu search algorithm for finding {E}($s^2$)-optimal
  and minimax-optimal supersaturated designs.
\newblock {\em Arxiv}, 2023.
\newblock \url{https://arxiv.org/abs/2303.09104}.

\bibitem{morales19}
L.~B. Morales, D.~A. Bulutoglu, and K.~T Arasu.
\newblock The maximum number of columns in supersaturated designs with
  $\smax=2$.
\newblock {\em J. Comb. Des.}, \textbf{27}:448--472, 2019.

\bibitem{morales2014}
L.~B. Morales and G.~Vega.
\newblock On the enumeration of $\es$-optimal and minimax-optimal $k$-circulant
  supersaturated designs.
\newblock {\em J. Comb. Des.}, \textbf{22}:149--160, 2014.

\bibitem{morales}
L.~B. Morales and C.~Velarde.
\newblock Enumeration of resolvable $2-(10,5,16)$ and $3-(10,5,6)$ designs.
\newblock {\em J. Comb. Des.}, \textbf{13}:108--119, 2005.

\bibitem{ngyuen}
N-K. Nguyen.
\newblock An algorithmic approach to constructing supersaturated designs.
\newblock {\em Technometrics}, \textbf{38}:69--73, 1996.

\bibitem{ngyuencheng}
N-K. Nguyen and C.~S. Cheng.
\newblock New $\es$-optimal supersaturated designs constructed from incomplete
  block designs.
\newblock {\em Technometrics}, \textbf{50}:26--31, 2008.

\bibitem{nurmela1997}
K.~J. Nurmela, M.~K. Kaikkonen, and P.~R.~J. {\"Osterg\aa rd}.
\newblock New constant weight codes from linear permutation groups.
\newblock {\em IEEE Trans. Inform. Theory}, \textbf{43}:1623--1630, 1997.

\bibitem{lcgc}
R.~Rodil, E.~Mart\'{i}nez, A.~M. Carro, R.~A. Lorenzo, and R.~Cela.
\newblock Applying supersaturated experimental designs to the study of
  composite sampling for monitoring pesticide residues in water.
\newblock {\em LCGC North America}, \textbf{22}:272--286, 2004.

\bibitem{ryan}
K.~J. Ryan and D.~A. Bulutoglu.
\newblock $\es$-optimal supersaturated designs with good minimax properties.
\newblock {\em Journal of Statistical Planning and Inference},
  \textbf{137}:2250--2262, 2007.

\bibitem{segunda}
P.~S. Segundo, D.~R. Losada, and A.~Jim\'enez.
\newblock An exact bit-parallel algorithm for the maximum clique problem.
\newblock {\em Comput. \& Oper. Res.}, \textbf{38}:571--581, 2011.

\bibitem{tang}
B.~Tang and C.~F.~J. Wu.
\newblock A method for constructing supersaturated designs and $\es$
  optimality.
\newblock {\em Can. J. Stat.}, \textbf{25}:191--201, 1997.

\bibitem{Williams68}
K.~R. Williams.
\newblock Designed experiments.
\newblock {\em Rubber Age}, \textbf{100}:65--71, 1968.

\bibitem{wu}
C.~F.~J. Wu.
\newblock Construction of supersaturated designs through partially alliased
  interactions.
\newblock {\em Biometrika}, \textbf{80}:661--669, 1993.

\bibitem{WuBook}
C.~F.~J. Wu and M.~Hamada.
\newblock {\em Experiments Planning, Analysis, and Parameter Design
  Optimization}.
\newblock Wiley, New York, NY, USA, 2000.

\end{thebibliography}
\newpage
\begin{appendices}
\section{ }\label{appendix}
\end{appendices}
\def\cD{{\bf D}}
\centerline{E{\bf $(s^2)$-optimal and minimax-optimal SSDs obtained by the TS algorithm}}
\begin{table}[ht]
	\begin{center}
\begin{tabular}{r@{\hskip 0.05in} r@{\hskip 0.05in} r@{\hskip 0.05in} r@{\hskip 0.05in} r@{\hskip 0.05in} r@{\hskip 0.05in} r@{\hskip 0.05in} r@{\hskip 0.05in} r@{\hskip 0.05in} r@{\hskip 0.05in} r@{\hskip 0.05in} r@{\hskip 0.05in} r@{\hskip 0.05in} r@{\hskip 0.05in} r@{\hskip 0.05in} r@{\hskip 0.05in} r@{\hskip 0.05in} r@{\hskip 0.05in} r@{\hskip 0.05in} r@{\hskip 0.05in} r@{\hskip 0.05in} r@{\hskip 0.05in} r@{\hskip 0.05in} r@{\hskip 0.05in} r@{\hskip 0.05in} r@{\hskip 0.05in} }
			\hline
			-1& 1&-1&-1&-1&-1& 1&-1& 1&-1& 1& 1&-1& 1&-1& 1& 1& 1& 1&-1&-1&-1&-1&-1& 1\\
			1&-1& 1&-1&-1&-1& 1&-1&-1& 1&-1&-1&-1& 1&-1&-1& 1&-1& 1&-1& 1& 1&-1& 1&-1\\
			-1& 1& 1& 1& 1&-1&-1&-1& 1&-1&-1&-1& 1&-1& 1&-1& 1& 1&-1&-1& 1& 1&-1& 1& 1\\
			1&-1& 1&-1& 1& 1& 1&-1& 1& 1& 1& 1&-1&-1& 1& 1&-1& 1& 1& 1& 1& 1& 1& 1& 1\\
			-1& 1&-1& 1&-1& 1& 1&-1& 1& 1&-1&-1& 1&-1& 1& 1&-1&-1& 1&-1& 1&-1& 1&-1&-1\\
			1& 1&-1& 1&-1&-1&-1& 1&-1& 1& 1&-1&-1&-1&-1& 1& 1&-1&-1& 1& 1& 1& 1&-1& 1\\
			-1& 1&-1&-1&-1& 1&-1& 1&-1& 1&-1& 1&-1& 1& 1&-1&-1& 1&-1&-1&-1& 1& 1& 1& 1\\
			-1&-1& 1&-1&-1&-1& 1& 1& 1& 1& 1& 1& 1&-1& 1&-1&-1&-1&-1& 1&-1& 1&-1&-1&-1\\
			1&-1&-1&-1&-1& 1&-1& 1&-1&-1&-1&-1& 1&-1& 1& 1& 1&-1& 1& 1&-1&-1&-1& 1& 1\\
			-1& 1& 1& 1&-1& 1& 1& 1&-1&-1& 1& 1& 1& 1&-1&-1& 1& 1& 1& 1& 1&-1& 1& 1&-1\\
			1&-1&-1& 1& 1&-1&-1& 1& 1& 1& 1& 1& 1& 1&-1&-1&-1&-1& 1&-1& 1&-1&-1& 1& 1\\
			1&-1& 1& 1& 1&-1&-1&-1&-1& 1&-1& 1& 1& 1& 1& 1& 1& 1&-1&-1&-1&-1& 1&-1&-1\\
			-1&-1&-1&-1& 1& 1&-1&-1& 1&-1& 1&-1&-1& 1& 1&-1& 1&-1&-1& 1& 1&-1& 1&-1&-1\\
			-1&-1& 1& 1& 1& 1&-1& 1&-1&-1& 1&-1&-1&-1&-1& 1&-1& 1& 1&-1&-1& 1&-1&-1&-1\\
			1& 1&-1&-1& 1&-1& 1& 1& 1&-1&-1&-1& 1& 1&-1& 1&-1& 1&-1& 1&-1& 1& 1& 1&-1\\
			1& 1& 1& 1& 1& 1& 1&-1&-1&-1&-1& 1&-1&-1&-1&-1&-1&-1&-1& 1&-1&-1&-1&-1& 1\\
			\hline
		\end{tabular}
		\caption{An E$(s^2)$-optimal, minimax-optimal SSD of size $N=16$, $m=25$, E$(s^2)= 7.68000$}
	\end{center}
\end{table}


\begin{table}[ht]
	\begin{center}
		\begin{tabular}{r@{\hskip 0.05in} r@{\hskip 0.05in} r@{\hskip 0.05in} r@{\hskip 0.05in} r@{\hskip 0.05in} r@{\hskip 0.05in} r@{\hskip 0.05in} r@{\hskip 0.05in} r@{\hskip 0.05in} r@{\hskip 0.05in} r@{\hskip 0.05in} r@{\hskip 0.05in} r@{\hskip 0.05in} r@{\hskip 0.05in} r@{\hskip 0.05in} r@{\hskip 0.05in} r@{\hskip 0.05in} r@{\hskip 0.05in} r@{\hskip 0.05in} r@{\hskip 0.05in} r@{\hskip 0.05in} r@{\hskip 0.05in} r@{\hskip 0.05in} r@{\hskip 0.05in} r@{\hskip 0.05in} r@{\hskip 0.05in} }
			\hline
			16& -4&  0&  0&  4& -4&  0&  0& -4&  4& -4&  0&  0&  0& -4&  4&  0& -4&  0&  4&  0&  0&  0&  4&  4\\
			-4& 16& -4&  4& -4&  0&  4&  0&  0& -4& -4&  0&  0&  0& -4&  0&  0&  4& -4&  0&  0&  0&  4&  0&  4\\
			0& -4& 16&  4&  4&  0&  4& -4& -4&  0&  0&  4&  0& -4&  0& -4&  0&  4&  0&  0&  0&  4& -4&  0& -4\\
			0&  4&  4& 16&  4&  0& -4&  0& -4&  0&  0&  0&  4& -4& -4&  0&  0&  0&  0& -4&  4& -4&  0& -4&  0\\
			4& -4&  4&  4& 16&  0& -4& -4&  4& -4&  0&  0&  0&  0&  0&  0& -4&  4& -4&  0&  0&  0&  0&  0&  0\\
			-4&  0&  0&  0&  0& 16&  0&  0& -4& -4&  0&  0& -4& -4&  4&  0& -4&  0&  4&  4&  0& -4&  4&  0&  0\\
			0&  4&  4& -4& -4&  0& 16& -4&  4&  0&  0&  4&  0&  0& -4&  0& -4&  0&  4&  4&  0&  0&  0&  0& -4\\
			0&  0& -4&  0& -4&  0& -4& 16& -4&  0&  4&  0&  4&  0& -4&  0& -4&  0&  0&  4& -4&  4&  0&  4&  0\\
\hline
		\end{tabular}
		\caption{$\cD^{\top}\cD$, $\smax=4$, $\freq=144$ }
	\end{center}
\end{table}

\begin{table}[ht]
	\begin{center}
		\begin{tabular}{r@{\hskip 0.05in} r@{\hskip 0.05in} r@{\hskip 0.05in} r@{\hskip 0.05in} r@{\hskip 0.05in} r@{\hskip 0.05in} r@{\hskip 0.05in} r@{\hskip 0.05in} r@{\hskip 0.05in} r@{\hskip 0.05in} r@{\hskip 0.05in} r@{\hskip 0.05in} r@{\hskip 0.05in} r@{\hskip 0.05in} r@{\hskip 0.05in} r@{\hskip 0.05in} r@{\hskip 0.05in} r@{\hskip 0.05in} r@{\hskip 0.05in} r@{\hskip 0.05in} r@{\hskip 0.05in} r@{\hskip 0.05in} r@{\hskip 0.05in} r@{\hskip 0.05in} r@{\hskip 0.05in} r@{\hskip 0.05in} }
			\hline
			1&-1&-1&-1&-1&-1&-1&-1&-1& 1&-1&-1&-1& 1& 1&-1& 1& 1&-1&-1&-1&-1&-1& 1& 1&-1\\
			1&-1& 1& 1&-1&-1& 1&-1& 1&-1& 1& 1& 1&-1&-1&-1& 1& 1&-1& 1& 1& 1&-1&-1& 1&-1\\
			-1& 1& 1& 1& 1& 1&-1& 1& 1& 1& 1&-1& 1& 1&-1&-1&-1& 1&-1& 1&-1&-1& 1& 1& 1&-1\\
			1& 1&-1& 1& 1& 1&-1& 1&-1& 1&-1& 1&-1&-1& 1&-1&-1& 1&-1& 1& 1& 1&-1&-1&-1& 1\\
			1&-1& 1&-1& 1& 1& 1& 1& 1& 1& 1&-1&-1&-1&-1& 1&-1& 1& 1&-1&-1&-1&-1&-1&-1&-1\\
			-1&-1&-1& 1&-1&-1& 1&-1&-1&-1& 1&-1&-1& 1&-1&-1&-1& 1& 1& 1&-1& 1& 1&-1&-1& 1\\
			-1& 1&-1&-1&-1&-1&-1& 1& 1& 1& 1&-1&-1&-1&-1& 1& 1&-1& 1& 1& 1& 1&-1& 1& 1& 1\\
			-1& 1& 1& 1& 1&-1&-1& 1&-1&-1&-1& 1&-1& 1&-1& 1& 1&-1&-1&-1&-1& 1&-1&-1&-1&-1\\
			1&-1&-1& 1& 1& 1& 1&-1&-1& 1&-1&-1& 1&-1&-1& 1& 1&-1&-1&-1&-1& 1& 1& 1& 1& 1\\
			1& 1& 1& 1&-1&-1& 1& 1&-1& 1& 1& 1& 1& 1& 1& 1& 1& 1& 1&-1& 1&-1& 1& 1&-1& 1\\
			-1& 1& 1&-1&-1& 1& 1&-1& 1& 1&-1& 1& 1& 1& 1&-1&-1&-1& 1&-1&-1& 1&-1&-1& 1& 1\\
			-1&-1&-1& 1& 1&-1& 1& 1& 1&-1&-1& 1&-1&-1& 1&-1&-1&-1& 1&-1& 1&-1& 1& 1& 1&-1\\
			-1&-1&-1&-1& 1& 1&-1&-1& 1&-1&-1& 1& 1& 1&-1& 1& 1& 1& 1& 1& 1&-1&-1& 1&-1& 1\\
			1& 1& 1&-1&-1& 1& 1&-1& 1&-1&-1&-1&-1& 1& 1& 1&-1&-1&-1& 1& 1& 1& 1& 1&-1&-1\\
			-1&-1& 1&-1&-1& 1&-1& 1&-1&-1& 1&-1& 1&-1& 1&-1& 1&-1&-1&-1& 1&-1& 1&-1&-1& 1\\
			1& 1&-1&-1& 1&-1&-1&-1&-1&-1& 1& 1& 1&-1& 1& 1&-1&-1& 1& 1&-1&-1& 1&-1& 1&-1\\
			\hline
		\end{tabular}
		\caption{An E$(s^2)$-optimal, minimax-optimal SSD of size $N=16$, $m=26$,  E$(s^2)= 7.87692$}
	\end{center}
\end{table}

\begin{table}[ht]
	\begin{center}
		\begin{tabular}{r@{\hskip 0.05in} r@{\hskip 0.05in} r@{\hskip 0.05in} r@{\hskip 0.05in} r@{\hskip 0.05in} r@{\hskip 0.05in} r@{\hskip 0.05in} r@{\hskip 0.05in} r@{\hskip 0.05in} r@{\hskip 0.05in} r@{\hskip 0.05in} r@{\hskip 0.05in} r@{\hskip 0.05in} r@{\hskip 0.05in} r@{\hskip 0.05in} r@{\hskip 0.05in} r@{\hskip 0.05in} r@{\hskip 0.05in} r@{\hskip 0.05in} r@{\hskip 0.05in} r@{\hskip 0.05in} r@{\hskip 0.05in} r@{\hskip 0.05in} r@{\hskip 0.05in} r@{\hskip 0.05in} r@{\hskip 0.05in} }
			\hline
			16&  0&  0&  0&  0&  0&  4& -4& -4&  4&  0&  0&  0& -4&  4&  4&  0&  4& -4&  0&  0&  0&  0&  0&  0& -4\\
			0& 16&  4&  0&  0&  0& -4&  4&  0&  4&  0&  4&  0&  4&  4&  4& -4& -4&  0&  4&  0&  4&  0&  0&  0&  0\\
			0&  4& 16&  0& -4&  4&  4&  4&  4&  0&  4&  0&  4&  4&  0&  0&  0&  0& -4& -4&  0&  0&  0& -4& -4& -4\\
			0&  0&  0& 16&  4& -4&  4&  4& -4&  0&  0&  4&  0&  0& -4& -4&  0&  4& -4&  0&  0&  4&  4&  0&  0&  0\\
			0&  0& -4&  4& 16&  4& -4&  4&  0&  0& -4&  4&  0& -4& -4&  4& -4&  0&  0&  0& -4& -4&  0&  0&  0& -4\\
			0&  0&  4& -4&  4& 16&  0&  0&  4&  4& -4& -4&  4&  0&  0&  0& -4&  0& -4&  0&  0&  0&  0&  0& -4&  4\\
			4& -4&  4&  4& -4&  0& 16& -4&  4&  0&  0&  0&  0&  0&  0&  0& -4&  0&  4& -4&  0&  4&  4&  0&  0&  0\\
			-4&  4&  4&  4&  4&  0& -4& 16&  0&  4&  4&  0& -4& -4&  0&  0&  0&  0&  0& -4&  4& -4&  0&  0& -4&  0\\
\hline
		\end{tabular}
		\caption{$\cD^{\top}\cD$,  $\smax=4$, $\freq=160$}
	\end{center}
\end{table}


\begin{table}[ht]
	\begin{center}
		\begin{tabular}{r@{\hskip 0.05in} r@{\hskip 0.05in} r@{\hskip 0.05in} r@{\hskip 0.05in} r@{\hskip 0.05in} r@{\hskip 0.05in} r@{\hskip 0.05in} r@{\hskip 0.05in} r@{\hskip 0.05in} r@{\hskip 0.05in} r@{\hskip 0.05in} r@{\hskip 0.05in} r@{\hskip 0.05in} r@{\hskip 0.05in} r@{\hskip 0.05in} r@{\hskip 0.05in} r@{\hskip 0.05in} r@{\hskip 0.05in} r@{\hskip 0.05in} r@{\hskip 0.05in} r@{\hskip 0.05in} r@{\hskip 0.05in} r@{\hskip 0.05in} r@{\hskip 0.05in} r@{\hskip 0.05in} r@{\hskip 0.05in} r@{\hskip 0.05in} }
			\hline
			-1&-1&-1& 1&-1& 1&-1& 1& 1& 1&-1&-1& 1&-1&-1& 1& 1&-1&-1& 1&-1&-1& 1&-1& 1&-1&-1\\
			-1& 1& 1&-1& 1&-1&-1& 1&-1&-1&-1& 1&-1&-1&-1& 1& 1&-1& 1& 1&-1& 1&-1& 1& 1& 1&-1\\
			1& 1&-1& 1& 1& 1& 1&-1&-1& 1& 1& 1&-1& 1&-1&-1& 1& 1&-1& 1&-1& 1& 1&-1&-1& 1&-1\\
			-1&-1&-1&-1& 1&-1& 1& 1& 1& 1& 1&-1& 1&-1& 1& 1& 1& 1& 1&-1& 1& 1&-1&-1&-1& 1&-1\\
			-1&-1&-1&-1&-1& 1& 1&-1&-1&-1&-1& 1& 1& 1& 1&-1&-1& 1&-1&-1&-1&-1&-1& 1& 1& 1&-1\\
			-1& 1&-1&-1&-1&-1& 1& 1&-1& 1&-1& 1&-1&-1&-1&-1&-1&-1&-1&-1& 1&-1& 1&-1&-1&-1& 1\\
			-1& 1& 1& 1& 1& 1&-1& 1&-1&-1& 1&-1& 1&-1& 1&-1&-1& 1& 1& 1&-1&-1&-1&-1&-1&-1& 1\\
			1&-1& 1& 1&-1&-1&-1& 1&-1& 1&-1&-1&-1& 1& 1& 1&-1& 1&-1& 1& 1& 1&-1& 1&-1&-1&-1\\
			1&-1& 1&-1& 1&-1&-1&-1& 1&-1&-1& 1&-1& 1& 1& 1& 1& 1&-1&-1&-1&-1& 1&-1&-1&-1& 1\\
			1&-1& 1&-1& 1& 1& 1&-1& 1&-1&-1&-1&-1&-1&-1&-1&-1&-1& 1& 1& 1&-1& 1& 1&-1& 1&-1\\
			1& 1&-1& 1&-1& 1&-1& 1& 1&-1&-1& 1& 1& 1& 1&-1& 1&-1& 1&-1& 1& 1& 1& 1&-1& 1& 1\\
			1& 1&-1&-1&-1&-1& 1& 1& 1&-1& 1&-1& 1& 1&-1& 1&-1& 1& 1& 1&-1& 1& 1& 1& 1&-1& 1\\
			1& 1& 1& 1&-1& 1& 1&-1& 1& 1& 1& 1&-1&-1& 1& 1&-1&-1& 1&-1&-1& 1&-1&-1& 1&-1&-1\\
			1& 1& 1&-1&-1& 1&-1&-1&-1& 1& 1&-1& 1&-1&-1& 1& 1& 1&-1&-1& 1&-1&-1& 1& 1& 1& 1\\
			-1&-1&-1& 1& 1&-1&-1&-1& 1&-1& 1&-1&-1& 1&-1&-1&-1&-1&-1&-1& 1& 1&-1&-1& 1& 1& 1\\
			-1&-1& 1& 1& 1&-1& 1&-1&-1& 1& 1& 1& 1& 1& 1&-1& 1&-1& 1& 1& 1&-1& 1& 1& 1&-1& 1\\
			\hline
		\end{tabular}
		\caption{An E$(s^2)$-optimal, minimax-optimal SSD of size $N=16$, $m=27$, E$(s^2) = 8.38746$}
	\end{center}
\end{table}
\begin{table}[ht]
	\begin{center}
		\begin{tabular}{r@{\hskip 0.05in} r@{\hskip 0.05in} r@{\hskip 0.05in} r@{\hskip 0.05in} r@{\hskip 0.05in} r@{\hskip 0.05in} r@{\hskip 0.05in} r@{\hskip 0.05in} r@{\hskip 0.05in} r@{\hskip 0.05in} r@{\hskip 0.05in} r@{\hskip 0.05in} r@{\hskip 0.05in} r@{\hskip 0.05in} r@{\hskip 0.05in} r@{\hskip 0.05in} r@{\hskip 0.05in} r@{\hskip 0.05in} r@{\hskip 0.05in} r@{\hskip 0.05in} r@{\hskip 0.05in} r@{\hskip 0.05in} r@{\hskip 0.05in} r@{\hskip 0.05in} r@{\hskip 0.05in} r@{\hskip 0.05in} r@{\hskip 0.05in} }
			\hline
			16&  4&  4&  0& -4&  4&  0& -4&  4&  0&  0&  0& -4&  4&  0&  4&  0&  4&  0&  0&  0&  4&  4&  4& -4&  0&  0\\
			4& 16&  0&  0& -4&  4&  0&  4& -4&  0&  4&  4&  0& -4& -4&  0&  0&  0&  4&  0& -4&  4&  0&  0&  0&  0&  4\\
			4&  0& 16&  0&  4&  0& -4& -4& -4&  0&  0&  0& -4& -4&  4&  4&  0&  0&  4&  4&  0& -4& -4&  4&  0& -4&  0\\
			0&  0&  0& 16&  0&  4& -4&  0&  0&  4&  4&  0&  0&  4&  4& -4&  0& -4&  0&  4&  0&  4&  0& -4&  0& -4&  0\\
			-4& -4&  4&  0& 16& -4&  0& -4&  0& -4&  4&  0& -4&  0&  0& -4&  4&  0&  4&  4&  0&  0&  0& -4& -4&  4&  0\\
			4&  4&  0&  4& -4& 16&  0& -4&  0&  0&  0&  0&  4& -4&  0& -4&  0&  0&  0&  0& -4& -4&  0&  0&  0&  4& -4\\
			0&  0& -4& -4&  0&  0& 16& -4&  0&  4&  4&  4&  0&  0&  0& -4& -4&  0&  4&  0&  0&  0&  4&  0&  0&  0& -4\\
			-4&  4& -4&  0& -4& -4& -4& 16&  0&  0& -4& -4&  4& -4&  0&  4&  0&  0&  4&  4&  0&  4&  0&  0& -4& -4&  0\\
\hline
		\end{tabular}
		\caption{$\cD^{\top}\cD$,  $\smax=4$, $\freq=184$}
	\end{center}
\end{table}


\begin{table}[ht]
	\begin{center}
		\begin{tabular}{r@{\hskip 0.05in} r@{\hskip 0.05in} r@{\hskip 0.05in} r@{\hskip 0.05in} r@{\hskip 0.05in} r@{\hskip 0.05in} r@{\hskip 0.05in} r@{\hskip 0.05in} r@{\hskip 0.05in} r@{\hskip 0.05in} r@{\hskip 0.05in} r@{\hskip 0.05in} r@{\hskip 0.05in} r@{\hskip 0.05in} r@{\hskip 0.05in} r@{\hskip 0.05in} r@{\hskip 0.05in} r@{\hskip 0.05in} r@{\hskip 0.05in} r@{\hskip 0.05in} r@{\hskip 0.05in} r@{\hskip 0.05in} r@{\hskip 0.05in} }
			\hline
			-1& 1& 1&-1& 1&-1&-1& 1&-1& 1&-1&-1& 1& 1& 1& 1&-1& 1& 1& 1& 1& 1&-1\\
			-1&-1&-1&-1& 1&-1& 1&-1& 1& 1&-1& 1&-1& 1&-1& 1&-1& 1& 1&-1&-1&-1& 1\\
			1&-1& 1&-1&-1& 1& 1&-1&-1& 1&-1&-1& 1&-1&-1& 1& 1&-1& 1&-1& 1&-1&-1\\
			-1& 1&-1&-1&-1&-1& 1&-1& 1& 1& 1& 1& 1&-1& 1&-1& 1&-1& 1& 1& 1& 1& 1\\
			1&-1& 1& 1& 1& 1&-1& 1& 1& 1& 1& 1&-1&-1& 1&-1&-1&-1& 1& 1& 1&-1& 1\\
			-1&-1&-1& 1& 1& 1& 1& 1&-1&-1&-1&-1& 1& 1& 1&-1&-1&-1&-1&-1& 1&-1& 1\\
			-1& 1&-1& 1&-1&-1&-1& 1&-1& 1& 1& 1&-1& 1&-1& 1& 1&-1&-1&-1& 1&-1&-1\\
			-1& 1& 1& 1&-1& 1&-1&-1& 1&-1&-1& 1& 1& 1& 1&-1& 1& 1& 1&-1&-1&-1&-1\\
			1& 1&-1& 1& 1&-1& 1&-1& 1&-1&-1&-1&-1&-1& 1& 1& 1& 1&-1& 1& 1&-1&-1\\
			1& 1& 1&-1& 1& 1& 1& 1&-1& 1& 1& 1&-1&-1& 1& 1& 1& 1&-1&-1&-1& 1& 1\\
			1& 1& 1&-1& 1&-1&-1&-1&-1&-1&-1& 1& 1&-1&-1&-1&-1&-1&-1& 1&-1&-1& 1\\
			1& 1& 1& 1&-1&-1& 1& 1& 1&-1& 1&-1& 1& 1&-1& 1&-1&-1& 1&-1&-1& 1& 1\\
			1&-1&-1& 1&-1&-1&-1&-1&-1& 1& 1&-1& 1&-1& 1&-1&-1& 1&-1&-1&-1& 1&-1\\
			1&-1&-1& 1& 1& 1&-1&-1&-1&-1& 1& 1& 1& 1&-1& 1& 1& 1& 1& 1& 1& 1& 1\\
			-1&-1& 1&-1&-1&-1&-1& 1& 1&-1&-1&-1&-1&-1&-1&-1& 1& 1&-1&-1& 1& 1& 1\\
			1&-1&-1&-1&-1& 1&-1& 1& 1&-1&-1& 1&-1& 1& 1& 1&-1&-1&-1& 1&-1& 1&-1\\
			-1&-1& 1& 1& 1& 1& 1&-1& 1& 1& 1&-1&-1& 1&-1&-1& 1&-1&-1& 1&-1& 1&-1\\
			-1& 1&-1&-1&-1& 1& 1& 1&-1&-1& 1&-1&-1&-1&-1&-1&-1& 1& 1& 1&-1&-1&-1\\
			\hline
		\end{tabular}
		\caption{An E$(s^2)$-optimal, minimax-optimal SSD of size $N=18$, $m=23$,  E$(s^2)= 6.15020$}
	\end{center}
\end{table}
\begin{table}[ht]
	\begin{center}
		\begin{tabular}{r@{\hskip 0.05in} r@{\hskip 0.05in} r@{\hskip 0.05in} r@{\hskip 0.05in} r@{\hskip 0.05in} r@{\hskip 0.05in} r@{\hskip 0.05in} r@{\hskip 0.05in} r@{\hskip 0.05in} r@{\hskip 0.05in} r@{\hskip 0.05in} r@{\hskip 0.05in} r@{\hskip 0.05in} r@{\hskip 0.05in} r@{\hskip 0.05in} r@{\hskip 0.05in} r@{\hskip 0.05in} r@{\hskip 0.05in} r@{\hskip 0.05in} r@{\hskip 0.05in} r@{\hskip 0.05in} r@{\hskip 0.05in} r@{\hskip 0.05in} }
			\hline
			18& -2&  2&  2&  2&  2& -2& -2& -2& -2&  2&  2&  2& -6&  2&  6& -2& -2& -2&  2& -2&  2&  2\\
			-2& 18&  2& -2& -2& -6&  2&  2& -2& -2&  2&  2&  2& -2&  2&  2&  2&  2&  2&  2& -2& -2& -2\\
			2&  2& 18& -2&  2&  2& -2&  2&  2&  2& -2& -2&  2& -2& -2& -2&  2& -2&  2& -2& -2&  2&  2\\
			2& -2& -2& 18&  2&  2& -2& -2&  2& -2&  6& -2&  2&  6&  2& -2&  2& -2& -2& -2&  2& -2& -2\\
			2& -2&  2&  2& 18&  2&  2& -2& -2&  2& -2&  2& -2&  2&  2&  2& -2&  2& -2&  6&  2& -2&  6\\
			2& -6&  2&  2&  2& 18&  2&  2& -2& -2&  2&  2& -2&  2&  2& -2&  2& -2&  2&  2& -2& -2& -2\\
			-2&  2& -2& -2&  2&  2& 18& -2&  2&  2&  2& -6& -2& -2& -2&  2&  2& -2&  2& -2& -2& -2&  2\\
			-2&  2&  2& -2& -2&  2& -2& 18& -2& -2&  2& -2& -6&  2&  2&  2& -6& -2& -2& -2&  2&  2&  2\\
\hline
		\end{tabular}
		\caption{$\cD^{\top}\cD$, $\smax=6$, $\freq=17$}
	\end{center}
\end{table}


\begin{table}[ht]
	\begin{center}
		\begin{tabular}{r@{\hskip 0.05in} r@{\hskip 0.05in} r@{\hskip 0.05in} r@{\hskip 0.05in} r@{\hskip 0.05in} r@{\hskip 0.05in} r@{\hskip 0.05in} r@{\hskip 0.05in} r@{\hskip 0.05in} r@{\hskip 0.05in} r@{\hskip 0.05in} r@{\hskip 0.05in} r@{\hskip 0.05in} r@{\hskip 0.05in} r@{\hskip 0.05in} r@{\hskip 0.05in} r@{\hskip 0.05in} r@{\hskip 0.05in} r@{\hskip 0.05in} r@{\hskip 0.05in} r@{\hskip 0.05in} r@{\hskip 0.05in} r@{\hskip 0.05in} r@{\hskip 0.05in} }
			\hline
			1&-1& 1& 1&-1& 1&-1&-1& 1& 1&-1&-1&-1&-1& 1& 1& 1& 1&-1&-1&-1& 1&-1& 1 \\
			-1&-1&-1&-1&-1&-1&-1& 1& 1& 1& 1& 1& 1&-1& 1& 1&-1&-1& 1& 1&-1&-1&-1& 1 \\
			1&-1& 1&-1&-1& 1& 1& 1&-1&-1&-1& 1&-1&-1&-1& 1&-1& 1& 1&-1&-1&-1& 1&-1 \\
			-1& 1& 1& 1& 1& 1&-1&-1&-1& 1& 1&-1&-1&-1& 1&-1&-1& 1& 1& 1& 1&-1&-1&-1 \\
			-1& 1& 1& 1& 1&-1&-1& 1&-1& 1&-1&-1& 1& 1&-1& 1&-1& 1&-1& 1&-1& 1& 1& 1 \\
			-1& 1&-1&-1&-1&-1& 1& 1&-1& 1& 1& 1&-1& 1& 1& 1& 1& 1&-1&-1& 1& 1&-1&-1 \\
			1& 1& 1&-1& 1&-1& 1&-1& 1& 1&-1& 1& 1&-1&-1& 1& 1&-1& 1& 1& 1& 1&-1&-1 \\
			1&-1&-1& 1& 1&-1& 1& 1& 1&-1& 1&-1& 1&-1&-1&-1& 1& 1&-1& 1&-1&-1&-1&-1 \\
			-1&-1&-1&-1& 1& 1& 1&-1& 1& 1& 1& 1&-1& 1&-1&-1& 1& 1& 1& 1&-1& 1& 1& 1 \\
			1&-1& 1&-1& 1&-1&-1&-1&-1&-1& 1& 1& 1&-1& 1&-1&-1& 1&-1&-1& 1& 1& 1& 1 \\
			1& 1& 1&-1&-1& 1&-1& 1&-1&-1& 1&-1& 1& 1& 1&-1& 1&-1& 1& 1&-1& 1& 1&-1 \\
			1& 1&-1& 1& 1& 1& 1&-1&-1&-1&-1& 1&-1& 1& 1& 1&-1&-1&-1& 1&-1&-1&-1& 1 \\
			1&-1&-1&-1&-1&-1&-1&-1&-1& 1&-1&-1&-1& 1&-1&-1&-1&-1&-1& 1& 1&-1& 1&-1 \\
			-1& 1&-1&-1& 1& 1&-1& 1& 1&-1& 1&-1&-1&-1&-1& 1& 1&-1&-1&-1& 1&-1& 1& 1 \\
			-1& 1&-1& 1&-1&-1& 1&-1& 1&-1&-1&-1&-1&-1& 1&-1&-1&-1& 1&-1&-1& 1& 1&-1 \\
			1& 1&-1& 1&-1& 1&-1& 1& 1&-1&-1& 1& 1& 1&-1&-1&-1& 1& 1&-1& 1& 1&-1& 1 \\
			-1&-1& 1& 1&-1&-1& 1&-1&-1&-1& 1&-1& 1& 1&-1& 1& 1&-1& 1&-1& 1&-1&-1& 1 \\
			-1&-1& 1& 1& 1& 1& 1& 1& 1& 1&-1& 1& 1& 1& 1&-1& 1&-1&-1&-1& 1&-1& 1&-1 \\
			\hline
		\end{tabular}
		\caption{An E$(s^2)$-optimal, minimax-optimal SSD of size $N=18$, $m=24$,  E$(s^2)= 6.66667$ }
	\end{center}
\end{table}

\begin{table}[ht]
	\begin{center}
		\begin{tabular}{r@{\hskip 0.05in} r@{\hskip 0.05in} r@{\hskip 0.05in} r@{\hskip 0.05in} r@{\hskip 0.05in} r@{\hskip 0.05in} r@{\hskip 0.05in} r@{\hskip 0.05in} r@{\hskip 0.05in} r@{\hskip 0.05in} r@{\hskip 0.05in} r@{\hskip 0.05in} r@{\hskip 0.05in} r@{\hskip 0.05in} r@{\hskip 0.05in} r@{\hskip 0.05in} r@{\hskip 0.05in} r@{\hskip 0.05in} r@{\hskip 0.05in} r@{\hskip 0.05in} r@{\hskip 0.05in} r@{\hskip 0.05in} r@{\hskip 0.05in} r@{\hskip 0.05in} }
			\hline
			18& -2&  2& -2& -2&  2& -2& -2& -2& -6& -6&  2&  2& -2& -2& -2& -2&  2& -2&  2& -2&  2& -2& -2\\
			-2& 18& -2&  2&  2&  2& -2&  2& -2& -2& -2& -2& -2&  2&  2&  2& -2& -2&  2&  2&  2&  6& -2& -2\\
			2& -2& 18&  2&  2&  2& -2& -2& -6&  2& -2& -2&  6& -2&  2&  2&  2&  2&  2& -2&  2&  2&  2& -2\\
			-2&  2&  2& 18&  2&  2&  2& -2&  2& -2& -6& -6&  2&  2&  2& -2& -2&  2& -2& -2& -2& -2& -6&  2\\
			-2&  2&  2&  2& 18&  2&  2& -2&  2&  2&  2&  2&  2& -2& -2& -2&  2&  2& -6&  6&  2& -2&  2&  2\\
			2&  2&  2&  2&  2& 18& -2&  2&  2& -2& -2&  2& -6&  2&  2& -2&  2&  2&  2& -2& -2& -2&  2&  2\\
			-2& -2& -2&  2&  2& -2& 18& -2&  2& -2& -2&  6& -2&  2& -2&  2&  6& -2&  2& -2& -2& -2& -2& -6\\
			-2&  2& -2& -2& -2&  2& -2& 18&  2& -2&  2&  2&  6&  2& -2&  2&  2&  2& -2& -2& -2& -2&  2& -2\\
\hline
		\end{tabular}
		\caption{$\cD^{\top}\cD$,  $\smax=6$, $\freq=23$}
	\end{center}
\end{table}

\begin{table}[ht]
	\begin{center}
		\begin{tabular}{r@{\hskip 0.05in} r@{\hskip 0.05in} r@{\hskip 0.05in} r@{\hskip 0.05in} r@{\hskip 0.05in} r@{\hskip 0.05in} r@{\hskip 0.05in} r@{\hskip 0.05in} r@{\hskip 0.05in} r@{\hskip 0.05in} r@{\hskip 0.05in} r@{\hskip 0.05in} r@{\hskip 0.05in} r@{\hskip 0.05in} r@{\hskip 0.05in} r@{\hskip 0.05in} r@{\hskip 0.05in} r@{\hskip 0.05in} r@{\hskip 0.05in} r@{\hskip 0.05in} r@{\hskip 0.05in} r@{\hskip 0.05in} r@{\hskip 0.05in} r@{\hskip 0.05in} r@{\hskip 0.05in} }
			\hline
			1& 1&-1& 1& 1&-1&-1&-1& 1& 1& 1&-1& 1& 1&-1&-1& 1&-1&-1&-1&-1& 1&-1&-1&-1\\
			1&-1& 1& 1&-1& 1& 1&-1&-1&-1&-1&-1& 1&-1& 1&-1&-1&-1&-1&-1& 1& 1& 1&-1& 1\\
			1&-1&-1& 1&-1& 1&-1& 1&-1& 1&-1&-1& 1& 1& 1& 1& 1& 1& 1& 1& 1& 1&-1& 1& 1\\
			1&-1& 1& 1&-1&-1& 1& 1& 1& 1& 1& 1&-1& 1&-1&-1& 1& 1&-1&-1&-1&-1& 1& 1& 1\\
			1&-1&-1&-1& 1&-1&-1& 1&-1&-1&-1& 1&-1& 1& 1&-1&-1&-1&-1& 1& 1& 1& 1& 1&-1\\
			1& 1& 1&-1&-1& 1&-1& 1& 1& 1&-1& 1& 1&-1&-1& 1&-1& 1&-1& 1&-1& 1& 1&-1&-1\\
			-1&-1&-1& 1& 1& 1&-1& 1& 1&-1& 1& 1& 1&-1& 1&-1& 1& 1& 1&-1& 1&-1& 1&-1&-1\\
			-1& 1& 1&-1&-1& 1&-1& 1&-1& 1& 1&-1&-1&-1& 1&-1&-1&-1& 1&-1&-1&-1&-1& 1&-1\\
			-1& 1&-1& 1& 1& 1& 1&-1& 1& 1&-1& 1&-1&-1& 1&-1& 1&-1& 1& 1&-1& 1& 1& 1& 1\\
			-1&-1& 1&-1& 1& 1& 1&-1& 1& 1&-1&-1&-1& 1&-1&-1&-1& 1& 1& 1& 1&-1&-1&-1&-1\\
			-1& 1&-1& 1&-1&-1& 1&-1&-1& 1& 1& 1&-1&-1& 1& 1&-1& 1&-1& 1& 1&-1&-1&-1&-1\\
			-1&-1&-1&-1&-1&-1& 1&-1& 1&-1& 1&-1& 1& 1& 1& 1&-1& 1& 1&-1&-1& 1& 1& 1&-1\\
			-1& 1& 1& 1& 1&-1&-1& 1& 1&-1&-1&-1&-1&-1&-1& 1&-1& 1&-1&-1& 1& 1&-1& 1& 1\\
			-1& 1& 1&-1&-1&-1& 1& 1&-1&-1& 1& 1& 1& 1&-1&-1& 1&-1& 1& 1& 1& 1&-1&-1& 1\\
			1& 1&-1& 1& 1& 1& 1& 1&-1&-1& 1&-1&-1& 1&-1& 1&-1&-1& 1& 1&-1&-1& 1&-1& 1\\
			-1&-1& 1&-1& 1&-1&-1&-1&-1& 1& 1&-1& 1&-1&-1& 1& 1&-1&-1& 1& 1&-1& 1& 1& 1\\
			1& 1& 1&-1& 1& 1& 1&-1&-1&-1&-1& 1& 1& 1& 1& 1& 1& 1&-1&-1&-1&-1&-1& 1&-1\\
			1&-1&-1&-1&-1&-1&-1&-1& 1&-1&-1& 1&-1&-1&-1& 1& 1&-1& 1&-1&-1&-1&-1&-1& 1\\
			\hline
		\end{tabular}
		\caption{An E$(s^2)$-optimal, minimax-optimal SSD of size $N=18$, $m=25$, E$(s^2)= 7.20000$}
	\end{center}
\end{table}

\begin{table}[ht]
	\begin{center}
		\begin{tabular}{r@{\hskip 0.05in} r@{\hskip 0.05in} r@{\hskip 0.05in} r@{\hskip 0.05in} r@{\hskip 0.05in} r@{\hskip 0.05in} r@{\hskip 0.05in} r@{\hskip 0.05in} r@{\hskip 0.05in} r@{\hskip 0.05in} r@{\hskip 0.05in} r@{\hskip 0.05in} r@{\hskip 0.05in} r@{\hskip 0.05in} r@{\hskip 0.05in} r@{\hskip 0.05in} r@{\hskip 0.05in} r@{\hskip 0.05in} r@{\hskip 0.05in} r@{\hskip 0.05in} r@{\hskip 0.05in} r@{\hskip 0.05in} r@{\hskip 0.05in} r@{\hskip 0.05in} r@{\hskip 0.05in} }
			\hline
			18& -2& -2&  2& -2&  2& -2&  2& -2& -2& -6&  2&  2&  6& -2&  2&  2& -2& -6& -2& -6&  2&  2& -2&  2\\
			-2& 18&  2&  2&  2&  2&  2&  2& -2&  2&  2&  2& -2& -2& -2&  2& -2& -2& -2&  2& -6&  2& -6& -2& -2\\
			-2&  2& 18& -6& -2&  2&  2&  2& -2&  2& -2& -2&  2& -2& -6& -2& -2&  2& -6& -2&  2& -2& -2&  2&  2\\
			2&  2& -6& 18&  2&  2&  2&  2&  2&  2&  2& -2& -2& -2&  2& -2&  2&  2& -2& -2&  2&  2&  2& -2&  6\\
			-2&  2& -2&  2& 18&  2& -2& -2&  2& -2& -2& -2& -2&  2& -2& -2&  2& -2& -2&  2&  2& -2&  2&  2& -2\\
			2&  2&  2&  2&  2& 18&  2&  2& -2&  2& -6& -2&  2& -2&  6& -2& -2&  2&  6&  2& -2& -2&  2& -2& -2\\
			-2&  2&  2&  2& -2&  2& 18& -6& -2& -2&  2&  2& -2&  6&  2& -2& -2&  2&  2&  2& -2& -2&  2& -2&  2\\
			2&  2&  2&  2& -2&  2& -6& 18& -2& -2&  2&  2& -2&  2& -2& -2& -2&  2&  2&  2&  2&  2&  2&  2&  2\\
\hline
		\end{tabular}
		\caption{$\cD^{\top}\cD$, $\smax=6$, $\freq=30$}
	\end{center}
\end{table}

\begin{table}[ht]
	\begin{center}
		\begin{tabular}{r@{\hskip 0.05in} r@{\hskip 0.05in} r@{\hskip 0.05in} r@{\hskip 0.05in} r@{\hskip 0.05in} r@{\hskip 0.05in} r@{\hskip 0.05in} r@{\hskip 0.05in} r@{\hskip 0.05in} r@{\hskip 0.05in} r@{\hskip 0.05in} r@{\hskip 0.05in} r@{\hskip 0.05in} r@{\hskip 0.05in} r@{\hskip 0.05in} r@{\hskip 0.05in} r@{\hskip 0.05in} r@{\hskip 0.05in} r@{\hskip 0.05in} r@{\hskip 0.05in} r@{\hskip 0.05in} r@{\hskip 0.05in} r@{\hskip 0.05in} r@{\hskip 0.05in} r@{\hskip 0.05in} r@{\hskip 0.05in} }
			\hline
			1& 1&-1&-1& 1&-1&-1&-1&-1&-1&-1& 1& 1&-1& 1& 1& 1& 1& 1&-1& 1&-1&-1&-1&-1&-1\\
			-1&-1& 1& 1& 1& 1& 1& 1&-1&-1&-1& 1& 1&-1& 1& 1&-1& 1& 1&-1&-1& 1& 1&-1&-1& 1\\
			-1& 1&-1& 1& 1& 1&-1&-1& 1& 1& 1& 1&-1&-1&-1&-1& 1& 1&-1&-1& 1& 1& 1&-1& 1& 1\\
			1& 1& 1&-1& 1& 1&-1& 1&-1& 1&-1&-1&-1& 1&-1&-1&-1& 1&-1&-1& 1& 1&-1& 1&-1& 1\\
			1& 1&-1& 1&-1&-1& 1& 1&-1& 1& 1& 1& 1& 1& 1& 1&-1& 1&-1& 1& 1& 1& 1& 1& 1&-1\\
			-1& 1&-1& 1& 1& 1&-1& 1& 1&-1& 1&-1&-1&-1& 1& 1&-1&-1& 1& 1& 1&-1&-1& 1&-1&-1\\
			1&-1& 1&-1&-1& 1& 1&-1& 1& 1& 1&-1& 1&-1& 1&-1&-1&-1&-1&-1& 1& 1&-1&-1&-1&-1\\
			-1& 1&-1&-1&-1&-1& 1& 1& 1&-1&-1&-1&-1&-1&-1&-1&-1& 1& 1& 1&-1& 1&-1&-1& 1&-1\\
\hline
		\end{tabular}
		\caption{An E$(s^2)$-optimal, minimax-optimal SSD of size $N=18$, $m=26$, E$(s^2)= 7.64308$}
	\end{center}
\end{table}

\begin{table}[ht]
	\begin{center}
		\begin{tabular}{r@{\hskip 0.05in} r@{\hskip 0.05in} r@{\hskip 0.05in} r@{\hskip 0.05in} r@{\hskip 0.05in} r@{\hskip 0.05in} r@{\hskip 0.05in} r@{\hskip 0.05in} r@{\hskip 0.05in} r@{\hskip 0.05in} r@{\hskip 0.05in} r@{\hskip 0.05in} r@{\hskip 0.05in} r@{\hskip 0.05in} r@{\hskip 0.05in} r@{\hskip 0.05in} r@{\hskip 0.05in} r@{\hskip 0.05in} r@{\hskip 0.05in} r@{\hskip 0.05in} r@{\hskip 0.05in} r@{\hskip 0.05in} r@{\hskip 0.05in} r@{\hskip 0.05in} r@{\hskip 0.05in} r@{\hskip 0.05in} }
			\hline
			18& -2& -2&  2& -2& -2&  2& -6& -6& -2&  2& -2&  2&  2& -2& -2& -2& -2& -2& -2&  6&  2& -2&  2& -2&  2\\
			-2& 18& -2& -2&  2& -2&  2&  2&  2& -2&  2&  2& -6& -2&  2&  2&  2&  6& -2&  2&  6&  2& -2&  2& -2&  2\\
			-2& -2& 18& -2&  2& -2&  6&  2&  2&  6& -2& -2&  2&  2&  2&  2&  2& -2& -2& -2&  2&  2& -2& -2& -2&  6\\
			2& -2& -2& 18&  2& -2& -2& -2&  2& -2&  2&  2&  2& -2&  2&  2& -2&  2& -2&  6&  2&  2&  2&  2&  2&  6\\
			-2&  2&  2&  2& 18&  2& -2&  2& -2& -2&  2& -2&  2& -2&  2&  2&  6&  2&  2& -6&  2&  2& -6&  2&  2&  2\\
			-2& -2& -2& -2&  2& 18& -2& -2& -6&  2&  2& -2& -6& -2&  2& -2& -2& -2& -2& -2& -2&  2&  2& -2& -6&  2\\
			2&  2&  6& -2& -2& -2& 18&  2& -2& -2&  6& -2&  2& -2&  2&  2&  2& -2&  2&  2& -2&  6&  2& -2&  2&  2\\
			-6&  2&  2& -2&  2& -2&  2& 18& -2&  2& -2&  2&  2& -2& -2& -2& -6& -2&  2&  2&  2& -2& -2&  2&  2& -2\\
\hline
		\end{tabular}
		\caption{$\cD^{\top}\cD$, $\smax=6$, $\freq=37$}
	\end{center}
\end{table}

\begin{table}[ht]
	\begin{center}
		\begin{tabular}{r@{\hskip 0.05in} r@{\hskip 0.05in} r@{\hskip 0.05in} r@{\hskip 0.05in} r@{\hskip 0.05in} r@{\hskip 0.05in} r@{\hskip 0.05in} r@{\hskip 0.05in} r@{\hskip 0.05in} r@{\hskip 0.05in} r@{\hskip 0.05in} r@{\hskip 0.05in} r@{\hskip 0.05in} r@{\hskip 0.05in} r@{\hskip 0.05in} r@{\hskip 0.05in} r@{\hskip 0.05in} r@{\hskip 0.05in} r@{\hskip 0.05in} r@{\hskip 0.05in} r@{\hskip 0.05in} r@{\hskip 0.05in} r@{\hskip 0.05in} r@{\hskip 0.05in} r@{\hskip 0.05in} r@{\hskip 0.05in} r@{\hskip 0.05in} }
			\hline
			1&-1&-1&-1& 1&-1& 1& 1& 1&-1&-1& 1& 1& 1&-1& 1& 1&-1& 1& 1& 1&-1& 1&-1&-1& 1&-1\\
			1& 1& 1&-1& 1&-1& 1&-1&-1& 1& 1& 1&-1&-1& 1& 1&-1& 1&-1&-1&-1& 1& 1&-1&-1& 1&-1\\
			-1& 1&-1& 1&-1& 1& 1&-1&-1&-1&-1& 1& 1& 1& 1& 1& 1& 1& 1&-1&-1& 1& 1&-1& 1&-1& 1\\
			-1&-1& 1& 1& 1& 1&-1&-1&-1& 1&-1& 1& 1&-1& 1& 1&-1& 1& 1& 1& 1&-1&-1& 1&-1& 1& 1\\
			-1& 1& 1&-1&-1&-1&-1&-1&-1&-1& 1& 1& 1&-1&-1&-1& 1&-1&-1&-1& 1&-1& 1& 1&-1&-1&-1\\
			-1&-1& 1&-1& 1& 1& 1& 1&-1&-1& 1&-1& 1& 1&-1&-1&-1& 1&-1& 1& 1& 1&-1&-1&-1&-1& 1\\
			1&-1&-1&-1&-1& 1&-1& 1&-1& 1&-1& 1&-1&-1&-1&-1&-1& 1&-1& 1&-1&-1& 1&-1& 1&-1&-1\\
			1&-1& 1&-1&-1& 1&-1& 1& 1&-1& 1& 1&-1& 1& 1&-1& 1& 1& 1&-1& 1& 1&-1& 1& 1& 1&-1\\
\hline
		\end{tabular}
		\caption{An E$(s^2)$-optimal, minimax-optimal SSD of size $N=18$, $m=27$, E$(s^2) = 8.10256$}
	\end{center}
\end{table}

\begin{table}[ht]
	\begin{center}
		\begin{tabular}{r@{\hskip 0.05in} r@{\hskip 0.05in} r@{\hskip 0.05in} r@{\hskip 0.05in} r@{\hskip 0.05in} r@{\hskip 0.05in} r@{\hskip 0.05in} r@{\hskip 0.05in} r@{\hskip 0.05in} r@{\hskip 0.05in} r@{\hskip 0.05in} r@{\hskip 0.05in} r@{\hskip 0.05in} r@{\hskip 0.05in} r@{\hskip 0.05in} r@{\hskip 0.05in} r@{\hskip 0.05in} r@{\hskip 0.05in} r@{\hskip 0.05in} r@{\hskip 0.05in} r@{\hskip 0.05in} r@{\hskip 0.05in} r@{\hskip 0.05in} r@{\hskip 0.05in} r@{\hskip 0.05in} r@{\hskip 0.05in} r@{\hskip 0.05in} }
			\hline
			18& -2& -6& -6& -2& -2&  2& -2&  2&  6&  2& -2& -2& -2& -2& -2&  2&  2&  2& -2&  2& -2&  2& -2&  2&  2& -2\\
			-2& 18&  2&  2& -2&  2&  6& -6& -2&  2& -2& -2& -6& -2& -2&  2&  2& -2&  2& -6& -2&  2&  2&  2& -2& -6& -2\\
			-6&  2& 18& -2& -2&  2& -2& -2&  2&  2&  2&  2& -2& -2&  6& -2& -2& -2& -2&  2&  6&  6& -2& -2& -6&  2& -2\\
			-6&  2& -2& 18&  2&  2& -2& -2&  2&  2& -2& -6&  2& -2&  2&  6&  6&  2&  2&  2& -2&  2&  2&  2&  2&  2&  2\\
			-2& -2& -2&  2& 18& -2&  2&  2& -2& -6&  2& -2&  2& -6& -2&  2& -2&  2&  2&  2& -2& -2& -6& -2& -2&  2& -2\\
			-2&  2&  2&  2& -2& 18& -2& -2& -2& -2& -2& -2& -6&  2& -2&  2& -2&  6&  2&  2&  2& -2& -2& -2&  2& -2&  2\\
			2&  6& -2& -2&  2& -2& 18& -2& -2&  2& -6& -2&  2&  6&  2& -2&  2&  2& -2&  2& -2&  2& -2& -2& -2&  2&  2\\
			-2& -6& -2& -2&  2& -2& -2& 18&  2& -2& -2&  2& -2&  6& -6& -2& -2&  2&  2&  2&  2&  6& -2&  2&  2& -2& -6\\
\hline
		\end{tabular}
		\caption{$\cD^{\top}\cD$,  $\smax=6$, $\freq=45$}
	\end{center}
\end{table}

\begin{table}[ht]
	\begin{center}
		\begin{tabular}{r@{\hskip 0.05in} r@{\hskip 0.05in} r@{\hskip 0.05in} r@{\hskip 0.05in} r@{\hskip 0.05in} r@{\hskip 0.05in} r@{\hskip 0.05in} r@{\hskip 0.05in} r@{\hskip 0.05in} r@{\hskip 0.05in} r@{\hskip 0.05in} r@{\hskip 0.05in} r@{\hskip 0.05in} r@{\hskip 0.05in} r@{\hskip 0.05in} r@{\hskip 0.05in} r@{\hskip 0.05in} r@{\hskip 0.05in} r@{\hskip 0.05in} r@{\hskip 0.05in} r@{\hskip 0.05in} r@{\hskip 0.05in} r@{\hskip 0.05in} r@{\hskip 0.05in} r@{\hskip 0.05in} r@{\hskip 0.05in} r@{\hskip 0.05in} r@{\hskip 0.05in} }
			\hline
			-1& 1&-1& 1&-1&-1& 1& 1&-1& 1&-1&-1& 1& 1& 1&-1& 1& 1&-1&-1&-1& 1&-1& 1& 1& 1& 1& 1\\
			1& 1& 1&-1& 1& 1&-1&-1&-1&-1& 1& 1&-1& 1& 1&-1& 1& 1& 1& 1&-1&-1&-1&-1& 1& 1& 1& 1\\
			-1&-1&-1& 1& 1&-1& 1& 1&-1&-1& 1& 1& 1&-1& 1& 1&-1& 1& 1& 1&-1&-1&-1& 1&-1&-1& 1&-1\\
			1&-1&-1& 1&-1& 1&-1&-1&-1&-1&-1&-1& 1&-1&-1& 1& 1&-1&-1& 1&-1& 1&-1&-1&-1& 1&-1& 1\\
			1&-1& 1&-1& 1&-1& 1&-1& 1& 1&-1&-1& 1&-1&-1& 1& 1& 1& 1&-1&-1& 1& 1&-1& 1&-1& 1&-1\\
			1& 1&-1& 1& 1& 1& 1&-1& 1&-1& 1& 1& 1& 1&-1&-1& 1& 1&-1&-1& 1& 1& 1& 1&-1&-1&-1& 1\\
			-1& 1&-1&-1&-1&-1&-1& 1& 1&-1& 1& 1& 1&-1&-1& 1& 1&-1& 1&-1& 1&-1& 1&-1& 1& 1& 1& 1\\
			-1&-1& 1& 1&-1&-1& 1&-1& 1&-1&-1&-1&-1& 1&-1&-1&-1& 1& 1&-1& 1&-1&-1&-1&-1& 1&-1&-1\\
\hline
		\end{tabular}
		\caption{An E$(s^2)$-optimal, minimax-optimal SSD of size $N=18$, $m=28$, E$(s^2)= 8.48677$}
	\end{center}
\end{table}
\begin{table}[ht]
	\begin{center}
		\begin{tabular}{r@{\hskip 0.05in} r@{\hskip 0.05in} r@{\hskip 0.05in} r@{\hskip 0.05in} r@{\hskip 0.05in} r@{\hskip 0.05in} r@{\hskip 0.05in} r@{\hskip 0.05in} r@{\hskip 0.05in} r@{\hskip 0.05in} r@{\hskip 0.05in} r@{\hskip 0.05in} r@{\hskip 0.05in} r@{\hskip 0.05in} r@{\hskip 0.05in} r@{\hskip 0.05in} r@{\hskip 0.05in} r@{\hskip 0.05in} r@{\hskip 0.05in} r@{\hskip 0.05in} r@{\hskip 0.05in} r@{\hskip 0.05in} r@{\hskip 0.05in} r@{\hskip 0.05in} r@{\hskip 0.05in} r@{\hskip 0.05in} r@{\hskip 0.05in} r@{\hskip 0.05in} }
			\hline
			18&  2&  2&  2& -2&  6& -2& -6&  2&  2& -2&  2&  2&  2&  2&  2&  2& -2& -2&  6&  2&  2&  2& -2& -2& -2&  2& -2\\
			2& 18& -2& -2& -2&  2& -2& -2& -2&  2&  2&  2& -2&  2& -2&  2&  6&  2& -2& -2&  2& -6&  2&  6&  2&  2&  2&  2\\
			2& -2& 18& -6& -2&  2&  2& -2&  2& -2& -2& -2& -2&  2&  2& -2& -2&  2& -2& -2& -2& -6&  2&  2&  2& -2& -2& -2\\
			2& -2& -6& 18& -2& -2&  2& -2& -2&  2& -2& -2&  2& -2&  6& -2& -2&  2& -6& -2&  2& -2& -2& -2& -6& -2& -2&  2\\
			-2& -2& -2& -2& 18&  2& -2&  2& -2&  2&  2&  2& -2&  2&  2&  2&  2&  6&  6& -2& -2&  2&  2& -2& -2& -6& -2&  2\\
			6&  2&  2& -2&  2& 18& -2& -2& -2&  2&  2&  2&  2&  2& -2& -2& -6&  2& -2&  2&  2&  2&  2&  2& -2&  6& -2&  6\\
			-2& -2&  2&  2& -2& -2& 18& -2&  2& -2& -2& -2&  2&  2& -2& -6& -2&  2&  2& -2& -6&  2&  6&  6& -2&  2&  2& -2\\
			-6& -2& -2& -2&  2& -2& -2& 18& -2&  2&  2& -2&  6&  6&  6&  2& -2& -2&  2&  2&  2& -2&  2&  2& -2& -2&  2&  2\\
\hline
		\end{tabular}
		\caption{$\cD^{\top}\cD$,  $\smax=6$, $\freq=53$}
	\end{center}
\end{table}

\begin{table}[ht]
	\begin{center}
		\begin{tabular}{r@{\hskip 0.05in} r@{\hskip 0.05in} r@{\hskip 0.05in} r@{\hskip 0.05in} r@{\hskip 0.05in} r@{\hskip 0.05in} r@{\hskip 0.05in} r@{\hskip 0.05in} r@{\hskip 0.05in} r@{\hskip 0.05in} r@{\hskip 0.05in} r@{\hskip 0.05in} r@{\hskip 0.05in} r@{\hskip 0.05in} r@{\hskip 0.05in} r@{\hskip 0.05in} r@{\hskip 0.05in} r@{\hskip 0.05in} r@{\hskip 0.05in} r@{\hskip 0.05in} r@{\hskip 0.05in} r@{\hskip 0.05in} r@{\hskip 0.05in} r@{\hskip 0.05in} r@{\hskip 0.05in} r@{\hskip 0.05in} r@{\hskip 0.05in} r@{\hskip 0.05in} r@{\hskip 0.05in} }
			\hline
			1&-1& 1& 1&-1& 1& 1&-1&-1&-1&-1&-1&-1& 1& 1& 1&-1&-1&-1&-1&-1& 1& 1& 1& 1&-1&-1&-1& 1\\
			-1& 1&-1& 1&-1&-1&-1&-1& 1& 1& 1&-1&-1&-1&-1& 1& 1&-1&-1&-1&-1&-1&-1&-1& 1&-1&-1&-1&-1\\
			-1&-1&-1&-1& 1&-1& 1&-1&-1& 1&-1& 1& 1&-1& 1&-1&-1& 1&-1& 1&-1& 1&-1& 1&-1&-1&-1&-1&-1\\
			1&-1&-1& 1& 1&-1& 1&-1&-1& 1&-1& 1&-1&-1&-1& 1&-1&-1& 1& 1& 1&-1& 1&-1& 1& 1& 1& 1& 1\\
			-1& 1& 1& 1& 1& 1& 1& 1&-1& 1&-1&-1& 1&-1&-1& 1& 1& 1&-1&-1& 1&-1& 1& 1&-1& 1&-1& 1&-1\\
			1&-1&-1&-1&-1& 1&-1& 1& 1& 1& 1& 1& 1&-1& 1& 1&-1& 1& 1&-1& 1& 1& 1&-1& 1&-1&-1& 1&-1\\
			1& 1& 1&-1& 1&-1& 1& 1& 1& 1& 1& 1& 1& 1& 1& 1& 1&-1&-1&-1&-1& 1& 1&-1&-1& 1& 1&-1& 1\\
			1& 1&-1&-1&-1&-1& 1&-1&-1&-1& 1&-1& 1&-1&-1&-1& 1& 1& 1&-1& 1& 1&-1& 1& 1& 1& 1&-1& 1\\
\hline
		\end{tabular}
		\caption{An E$(s^2)$-optimal, minimax-optimal SSD of size $N=18$, $m=29$, E$(s^2)= 8.72906$}
	\end{center}
\end{table}


\begin{table}[ht]
	\begin{center}
		\begin{tabular}{r@{\hskip 0.05in} r@{\hskip 0.05in} r@{\hskip 0.05in} r@{\hskip 0.05in} r@{\hskip 0.05in} r@{\hskip 0.05in} r@{\hskip 0.05in} r@{\hskip 0.05in} r@{\hskip 0.05in} r@{\hskip 0.05in} r@{\hskip 0.05in} r@{\hskip 0.05in} r@{\hskip 0.05in} r@{\hskip 0.05in} r@{\hskip 0.05in} r@{\hskip 0.05in} r@{\hskip 0.05in} r@{\hskip 0.05in} r@{\hskip 0.05in} r@{\hskip 0.05in} r@{\hskip 0.05in} r@{\hskip 0.05in} r@{\hskip 0.05in} r@{\hskip 0.05in} r@{\hskip 0.05in} r@{\hskip 0.05in} r@{\hskip 0.05in} r@{\hskip 0.05in} r@{\hskip 0.05in} }
			\hline
			18&  2&  2& -2& -2& -6& -2&  2& -2& -2& -2&  6&  2&  2& -2&  6& -2&  2&  2& -2&  2&  2&  2& -2&  2& -2&  2& -2&  6\\
			2& 18&  2&  2& -2& -2& -2&  2&  2&  2&  2& -2&  2&  2& -2&  2&  2&  6& -6& -2&  6& -2& -2&  2& -2&  2&  6& -6& -2\\
			2&  2& 18&  2&  2&  2& -2&  6& -6& -2& -2&  2& -2&  2&  2& -2&  2& -2& -2& -6& -2& -2&  2& -2& -2&  2&  2& -2& -2\\
			-2&  2&  2& 18& -2&  6& -2& -2& -2& -2& -6& -2& -2& -2& -6&  2&  2& -2& -6&  2&  2&  2&  2& -6&  2&  2& -2& -2&  2\\
			-2& -2&  2& -2& 18&  2&  2& -2& -2&  2& -2&  2&  6& -2& -2&  2&  2& -2&  2&  2&  2& -6&  2&  2& -6& -2&  2& -2&  2\\
			-6& -2&  2&  6&  2& 18& -2& -2&  2& -2&  2& -2&  2&  2&  2& -2& -2&  2& -2& -2&  6&  2&  6&  2&  2& -2& -2&  2&  2\\
			-2& -2& -2& -2&  2& -2& 18& -6& -2& -2& -2& -2& -2& -2&  2&  2& -2&  2& -2& -2& -2&  2&  6&  2& -6&  6&  2& -2&  2\\
			2&  2&  6& -2& -2& -2& -6& 18&  2&  2& -2& -2&  6&  2&  2&  2&  2& -2&  2&  2& -2&  2&  2& -2& -2&  2&  2&  6& -6\\
\hline
		\end{tabular}
		\caption{$\cD^{\top}\cD$, $\smax=6$, $\freq=60$}
	\end{center}
\end{table}

\begin{table}[ht]
	\begin{center}
		\begin{tabular}{r@{\hskip 0.05in} r@{\hskip 0.05in} r@{\hskip 0.05in} r@{\hskip 0.05in} r@{\hskip 0.05in} r@{\hskip 0.05in} r@{\hskip 0.05in} r@{\hskip 0.05in} r@{\hskip 0.05in} r@{\hskip 0.05in} r@{\hskip 0.05in} r@{\hskip 0.05in} r@{\hskip 0.05in} r@{\hskip 0.05in} r@{\hskip 0.05in} r@{\hskip 0.05in} r@{\hskip 0.05in} r@{\hskip 0.05in} r@{\hskip 0.05in} r@{\hskip 0.05in} r@{\hskip 0.05in} }
			\hline
			1&-1& 1&-1&-1& 1& 1& 1&-1& 1&-1& 1& 1&-1& 1&-1&-1&-1& 1& 1& 1\\
			-1&-1&-1& 1& 1&-1&-1& 1&-1& 1& 1& 1&-1&-1& 1&-1& 1&-1&-1&-1&-1\\
			1&-1& 1& 1& 1& 1& 1&-1& 1&-1& 1& 1&-1& 1& 1& 1& 1&-1& 1&-1& 1\\
			-1&-1& 1&-1& 1&-1& 1& 1& 1& 1& 1&-1&-1&-1&-1& 1&-1& 1&-1&-1& 1\\
			-1& 1&-1&-1&-1&-1& 1& 1& 1&-1& 1& 1& 1&-1& 1& 1& 1& 1& 1& 1&-1\\
			1& 1& 1&-1&-1& 1&-1&-1&-1& 1& 1&-1& 1&-1&-1& 1& 1&-1& 1&-1&-1\\
			-1&-1&-1& 1& 1& 1& 1&-1&-1&-1&-1& 1& 1&-1&-1& 1&-1& 1& 1&-1&-1\\
			1& 1& 1& 1& 1&-1&-1& 1&-1&-1&-1&-1&-1&-1& 1& 1&-1& 1& 1& 1& 1\\
			1& 1& 1& 1&-1&-1& 1& 1&-1&-1& 1& 1& 1& 1&-1&-1& 1& 1&-1&-1& 1\\
			-1&-1& 1&-1&-1& 1&-1&-1&-1&-1& 1&-1&-1& 1& 1&-1&-1& 1&-1& 1&-1\\
			1&-1&-1& 1& 1& 1&-1& 1& 1& 1& 1&-1& 1& 1&-1&-1& 1& 1& 1& 1&-1\\
			1& 1&-1& 1&-1& 1& 1&-1& 1& 1&-1&-1&-1&-1& 1&-1& 1& 1&-1&-1&-1\\
			1& 1&-1&-1& 1& 1&-1&-1& 1&-1& 1& 1&-1&-1&-1&-1&-1&-1&-1& 1& 1\\
			-1& 1& 1& 1&-1& 1& 1& 1& 1& 1&-1& 1&-1& 1&-1& 1&-1&-1&-1& 1&-1\\
			-1& 1& 1&-1& 1&-1&-1&-1& 1& 1&-1& 1& 1& 1& 1&-1&-1& 1& 1&-1&-1\\
			-1& 1&-1&-1& 1& 1& 1& 1&-1&-1&-1&-1&-1& 1&-1&-1& 1&-1& 1&-1& 1\\
			-1&-1& 1& 1&-1&-1&-1&-1& 1&-1&-1&-1& 1&-1&-1&-1& 1&-1&-1& 1& 1\\
			-1& 1&-1& 1& 1&-1& 1&-1&-1& 1& 1&-1& 1& 1& 1& 1&-1&-1&-1& 1& 1\\
			1&-1&-1&-1&-1&-1&-1& 1& 1&-1&-1&-1& 1& 1& 1& 1&-1&-1&-1&-1&-1\\
			1&-1&-1&-1&-1&-1&-1&-1&-1& 1&-1& 1&-1& 1&-1& 1& 1& 1& 1& 1& 1\\
			\hline
		\end{tabular}
		\caption{An E$(s^2)$-optimal, minimax-optimal SSD of size $N=20$, $m=21$, E$(s^2)= 3.80952$}
	\end{center}
\end{table}

\begin{table}[ht]
	\begin{center}
		\begin{tabular}{r@{\hskip 0.05in} r@{\hskip 0.05in} r@{\hskip 0.05in} r@{\hskip 0.05in} r@{\hskip 0.05in} r@{\hskip 0.05in} r@{\hskip 0.05in} r@{\hskip 0.05in} r@{\hskip 0.05in} r@{\hskip 0.05in} r@{\hskip 0.05in} r@{\hskip 0.05in} r@{\hskip 0.05in} r@{\hskip 0.05in} r@{\hskip 0.05in} r@{\hskip 0.05in} r@{\hskip 0.05in} r@{\hskip 0.05in} r@{\hskip 0.05in} r@{\hskip 0.05in} r@{\hskip 0.05in} }
			\hline
			20&  0&  0&  0& -4&  4& -4&  0&  0&  0&  0&  0&  0&  0&  0&  0&  4&  0&  4&  0&  4\\
			0& 20&  0&  0&  0&  0&  4&  0&  0&  0&  0&  0&  0&  0&  0&  0&  0&  0&  0&  0&  0\\
			0&  0& 20&  0& -4&  0&  0&  0&  0&  0&  0&  0&  0&  0&  0&  0& -4&  0&  0&  0&  4\\
			0&  0&  0& 20&  4&  0&  4&  0&  0&  0&  0&  0&  0&  0&  0&  0&  4&  0& -4&  0&  0\\
			-4&  0& -4&  4& 20&  0&  0&  0&  0&  0&  4&  0& -4&  0&  0&  0& -4&  0&  4& -4&  4\\
			4&  0&  0&  0&  0& 20&  4& -4&  0&  0&  0&  0& -4&  0& -4& -4&  0& -4&  4&  0& -4\\
			-4&  4&  0&  4&  0&  4& 20&  4&  0&  0&  0&  4&  0&  0&  0&  4&  0&  0&  0& -4&  4\\
			0&  0&  0&  0&  0& -4&  4& 20&  0&  0&  0&  0&  0&  0&  0&  0&  0&  0&  0&  0&  0\\
\hline
		\end{tabular}
		\caption{$\cD^{\top}\cD$,  $\smax=4$, $\freq=50$}
	\end{center}
\end{table}

\begin{table}[ht]
	\begin{center}
		\begin{tabular}{r@{\hskip 0.05in} r@{\hskip 0.05in} r@{\hskip 0.05in} r@{\hskip 0.05in} r@{\hskip 0.05in} r@{\hskip 0.05in} r@{\hskip 0.05in} r@{\hskip 0.05in} r@{\hskip 0.05in} r@{\hskip 0.05in} r@{\hskip 0.05in} r@{\hskip 0.05in} r@{\hskip 0.05in} r@{\hskip 0.05in} r@{\hskip 0.05in} r@{\hskip 0.05in} r@{\hskip 0.05in} r@{\hskip 0.05in} r@{\hskip 0.05in} r@{\hskip 0.05in} r@{\hskip 0.05in} r@{\hskip 0.05in} }
			\hline
			-1& 1& 1&-1& 1&-1&-1&-1& 1& 1&-1&-1& 1&-1& 1& 1&-1& 1& 1&-1& 1& 1\\
			-1& 1&-1&-1&-1&-1&-1& 1& 1&-1& 1& 1&-1& 1& 1&-1&-1& 1&-1&-1& 1& 1\\
			1& 1&-1&-1&-1&-1& 1&-1&-1& 1&-1& 1& 1&-1&-1&-1& 1&-1&-1&-1& 1&-1\\
			-1&-1& 1&-1& 1& 1&-1&-1&-1& 1& 1& 1& 1& 1&-1&-1& 1&-1&-1& 1& 1& 1\\
			1&-1& 1&-1&-1&-1& 1& 1& 1&-1&-1&-1&-1&-1&-1& 1&-1&-1&-1& 1& 1& 1\\
			1&-1& 1&-1&-1& 1&-1& 1&-1& 1&-1&-1&-1&-1& 1&-1& 1& 1&-1& 1&-1&-1\\
			1& 1& 1& 1&-1& 1&-1& 1&-1&-1& 1& 1& 1&-1& 1& 1& 1&-1&-1&-1& 1& 1\\
			-1&-1&-1& 1& 1&-1&-1& 1&-1&-1&-1& 1& 1&-1&-1& 1&-1& 1&-1& 1&-1&-1\\
			1& 1&-1&-1&-1& 1& 1& 1&-1& 1& 1&-1& 1& 1&-1& 1&-1& 1& 1&-1&-1& 1\\
			-1&-1&-1&-1& 1& 1& 1&-1& 1&-1& 1&-1& 1&-1& 1&-1&-1&-1&-1&-1&-1&-1\\
			-1&-1&-1& 1&-1& 1&-1& 1& 1& 1&-1&-1&-1&-1&-1&-1& 1&-1& 1&-1&-1& 1\\
			1& 1&-1& 1& 1& 1& 1&-1&-1&-1&-1&-1&-1& 1&-1&-1&-1& 1&-1& 1& 1& 1\\
			-1& 1& 1& 1&-1&-1&-1&-1&-1&-1& 1&-1&-1& 1&-1&-1&-1&-1& 1& 1&-1&-1\\
			1& 1& 1&-1& 1&-1& 1& 1& 1&-1&-1& 1& 1& 1& 1&-1& 1&-1& 1& 1&-1& 1\\
			-1&-1&-1& 1&-1&-1& 1&-1& 1& 1& 1&-1& 1& 1& 1& 1& 1& 1&-1& 1&-1& 1\\
			1&-1& 1& 1& 1&-1&-1& 1& 1&-1& 1&-1& 1& 1&-1&-1& 1& 1& 1&-1& 1&-1\\
			-1& 1&-1& 1& 1&-1& 1& 1&-1& 1& 1&-1&-1&-1& 1& 1& 1&-1& 1& 1& 1&-1\\
			1& 1&-1&-1& 1& 1&-1&-1& 1&-1& 1& 1&-1&-1&-1& 1& 1& 1& 1& 1&-1&-1\\
			1&-1&-1& 1&-1& 1&-1&-1& 1& 1&-1& 1& 1& 1& 1& 1&-1&-1& 1& 1& 1&-1\\
			1&-1& 1& 1& 1&-1& 1&-1&-1& 1& 1& 1&-1&-1& 1&-1&-1& 1& 1&-1&-1& 1\\
			-1&-1& 1&-1&-1& 1& 1&-1&-1&-1&-1& 1&-1& 1& 1& 1& 1& 1& 1&-1& 1&-1\\
			-1& 1& 1& 1& 1& 1& 1& 1& 1& 1&-1& 1&-1& 1&-1& 1&-1&-1&-1&-1&-1&-1\\
			\hline
		\end{tabular}
		\caption{An E$(s^2)$-optimal, minimax-optimal SSD of size $N=22$, $m=22$, E$(s^2)= 4.00000$}
	\end{center}
\end{table}

\begin{table}[ht]
	\begin{center}
		\begin{tabular}{r@{\hskip 0.05in} r@{\hskip 0.05in} r@{\hskip 0.05in} r@{\hskip 0.05in} r@{\hskip 0.05in} r@{\hskip 0.05in} r@{\hskip 0.05in} r@{\hskip 0.05in} r@{\hskip 0.05in} r@{\hskip 0.05in} r@{\hskip 0.05in} r@{\hskip 0.05in} r@{\hskip 0.05in} r@{\hskip 0.05in} r@{\hskip 0.05in} r@{\hskip 0.05in} r@{\hskip 0.05in} r@{\hskip 0.05in} r@{\hskip 0.05in} r@{\hskip 0.05in} r@{\hskip 0.05in} r@{\hskip 0.05in} }
			\hline
			22&  2&  2& -2& -2&  2&  2&  2& -2& -2& -2&  2&  2& -2& -2& -2&  2&  2&  2&  2&  2&  2\\
			2& 22& -2& -2&  2& -2&  2&  2& -2& -2&  2&  2& -2&  2& -2&  2& -2& -2&  2& -2&  2&  2\\
			2& -2& 22& -2&  2& -2& -2&  2& -2& -2& -2&  2& -2&  2&  2& -2&  2& -2&  2& -2&  2&  2\\
			-2& -2& -2& 22&  2& -2& -2&  2& -2&  2&  2& -2& -2&  2& -2&  2& -2& -2&  2&  2& -2& -2\\
			-2&  2&  2&  2& 22& -2&  2& -2&  2& -2&  2&  2&  2& -2& -2& -2& -2&  2&  2&  2& -2& -2\\
			2& -2& -2& -2& -2& 22& -2& -2& -2&  2& -2&  2& -2&  2& -2&  2&  2& -2& -2& -2& -2& -2\\
			2&  2& -2& -2&  2& -2& 22& -2& -2&  2& -2& -2& -2&  2&  2&  2& -2& -2& -2& -2& -2&  2\\
			2&  2&  2&  2& -2& -2& -2& 22&  2& -2& -2& -2& -2& -2& -2&  2&  2& -2& -2& -2& -2&  2\\
\hline
		\end{tabular}
		\caption{$\cD^{\top}\cD$,  $\smax=2$, $\freq=231$}
	\end{center}
\end{table}


\begin{table}[ht]
	\begin{center}
		\begin{tabular}{r@{\hskip 0.05in} r@{\hskip 0.05in} r@{\hskip 0.05in} r@{\hskip 0.05in} r@{\hskip 0.05in} r@{\hskip 0.05in} r@{\hskip 0.05in} r@{\hskip 0.05in} r@{\hskip 0.05in} r@{\hskip 0.05in} r@{\hskip 0.05in} r@{\hskip 0.05in} r@{\hskip 0.05in} r@{\hskip 0.05in} r@{\hskip 0.05in} r@{\hskip 0.05in} r@{\hskip 0.05in} r@{\hskip 0.05in} r@{\hskip 0.05in} r@{\hskip 0.05in} r@{\hskip 0.05in} r@{\hskip 0.05in} r@{\hskip 0.05in} }
			\hline
			-1& 1& 1& 1&-1& 1&-1& 1& 1&-1& 1& 1&-1& 1& 1&-1& 1& 1&-1& 1&-1& 1& 1\\
			1& 1& 1& 1&-1&-1& 1& 1&-1& 1& 1&-1&-1& 1&-1&-1& 1&-1&-1&-1& 1&-1&-1\\
			-1& 1& 1&-1& 1&-1& 1& 1&-1&-1& 1&-1& 1&-1& 1& 1&-1& 1&-1&-1& 1& 1& 1\\
			-1&-1&-1&-1&-1&-1&-1&-1&-1&-1& 1& 1& 1&-1&-1&-1& 1& 1&-1& 1& 1&-1&-1\\
			-1&-1& 1&-1&-1& 1& 1& 1& 1& 1& 1&-1& 1&-1& 1& 1&-1&-1& 1& 1&-1&-1&-1\\
			-1& 1&-1& 1& 1& 1& 1&-1& 1& 1& 1& 1& 1&-1&-1& 1& 1&-1&-1&-1&-1&-1& 1\\
			1& 1&-1&-1&-1&-1& 1& 1&-1& 1&-1& 1&-1&-1& 1& 1& 1& 1& 1& 1&-1&-1& 1\\
			1& 1& 1& 1&-1&-1&-1&-1&-1& 1&-1& 1& 1&-1& 1& 1&-1&-1&-1& 1&-1& 1&-1\\
\hline
		\end{tabular}
		\caption{An E$(s^2)$-optimal, minimax-optimal SSD of size $N=22$, $m=23$, E$(s^2)= 4.00000$}
	\end{center}
\end{table}
\begin{table}[ht]
	\begin{center}
		\begin{tabular}{r@{\hskip 0.05in} r@{\hskip 0.05in} r@{\hskip 0.05in} r@{\hskip 0.05in} r@{\hskip 0.05in} r@{\hskip 0.05in} r@{\hskip 0.05in} r@{\hskip 0.05in} r@{\hskip 0.05in} r@{\hskip 0.05in} r@{\hskip 0.05in} r@{\hskip 0.05in} r@{\hskip 0.05in} r@{\hskip 0.05in} r@{\hskip 0.05in} r@{\hskip 0.05in} r@{\hskip 0.05in} r@{\hskip 0.05in} r@{\hskip 0.05in} r@{\hskip 0.05in} r@{\hskip 0.05in} r@{\hskip 0.05in} r@{\hskip 0.05in} }
			\hline
			22& -2& -2& -2& -2& -2&  2&  2& -2&  2& -2&  2&  2&  2&  2& -2&  2& -2& -2& -2& -2&  2&  2\\
			-2& 22& -2& -2& -2& -2&  2&  2& -2&  2& -2&  2&  2&  2&  2& -2&  2& -2& -2& -2& -2&  2&  2\\
			-2& -2& 22&  2& -2& -2& -2& -2&  2&  2&  2& -2& -2&  2&  2&  2& -2&  2&  2& -2&  2& -2&  2\\
			-2& -2&  2& 22& -2&  2&  2& -2& -2&  2& -2&  2& -2&  2& -2& -2& -2&  2& -2& -2& -2&  2& -2\\
			-2& -2& -2& -2& 22&  2&  2& -2&  2&  2&  2& -2& -2&  2& -2& -2& -2&  2& -2& -2&  2&  2&  2\\
			-2& -2& -2&  2&  2& 22& -2&  2&  2& -2&  2& -2&  2&  2&  2&  2&  2& -2&  2& -2&  2& -2&  2\\
			2&  2& -2&  2&  2& -2& 22&  2&  2& -2&  2& -2&  2& -2&  2&  2&  2& -2&  2&  2&  2& -2&  2\\
			2&  2& -2& -2& -2&  2&  2& 22&  2&  2&  2& -2& -2& -2& -2&  2& -2&  2& -2&  2&  2&  2&  2\\
\hline
		\end{tabular}
		\caption{$\cD^{\top}\cD$, $\smax=2$, $\freq=253$}
	\end{center}
\end{table}


\begin{table}[ht]
	\begin{center}
		\begin{tabular}{r@{\hskip 0.05in} r@{\hskip 0.05in} r@{\hskip 0.05in} r@{\hskip 0.05in} r@{\hskip 0.05in} r@{\hskip 0.05in} r@{\hskip 0.05in} r@{\hskip 0.05in} r@{\hskip 0.05in} r@{\hskip 0.05in} r@{\hskip 0.05in} r@{\hskip 0.05in} r@{\hskip 0.05in} r@{\hskip 0.05in} r@{\hskip 0.05in} r@{\hskip 0.05in} r@{\hskip 0.05in} r@{\hskip 0.05in} r@{\hskip 0.05in} r@{\hskip 0.05in} r@{\hskip 0.05in} r@{\hskip 0.05in} r@{\hskip 0.05in} r@{\hskip 0.05in} }
			\hline
			-1& 1& 1& 1& 1&-1& 1& 1&-1& 1&-1& 1&-1& 1& 1& 1&-1& 1& 1&-1& 1& 1& 1& 1\\
			1& 1&-1& 1&-1&-1& 1& 1& 1&-1&-1&-1&-1& 1&-1&-1& 1&-1&-1&-1&-1& 1&-1& 1\\
			-1&-1&-1&-1& 1&-1& 1& 1& 1& 1& 1&-1&-1&-1&-1& 1&-1& 1&-1& 1&-1& 1&-1&-1\\
			-1& 1& 1& 1&-1&-1&-1&-1&-1& 1& 1& 1& 1&-1&-1&-1&-1&-1& 1&-1&-1& 1&-1&-1\\
			1& 1& 1&-1&-1& 1&-1& 1&-1&-1&-1& 1&-1&-1& 1& 1&-1& 1&-1&-1&-1&-1&-1&-1\\
			1&-1&-1& 1&-1&-1&-1&-1& 1&-1& 1& 1&-1& 1&-1& 1&-1& 1& 1& 1& 1&-1&-1& 1\\
			-1& 1& 1&-1&-1& 1& 1& 1& 1&-1& 1&-1& 1& 1&-1&-1&-1& 1& 1&-1& 1&-1& 1&-1\\
			1& 1& 1&-1&-1& 1& 1&-1& 1& 1&-1& 1& 1&-1&-1& 1&-1&-1&-1& 1& 1& 1& 1& 1\\
\hline
		\end{tabular}
		\caption{An E$(s^2)$-optimal, minimax-optimal SSD of size $N=24$, $m=24$,  E$(s^2)= 2.08696$}
	\end{center}
\end{table}

\begin{table}[ht]
	\begin{center}
		\begin{tabular}{r@{\hskip 0.05in} r@{\hskip 0.05in} r@{\hskip 0.05in} r@{\hskip 0.05in} r@{\hskip 0.05in} r@{\hskip 0.05in} r@{\hskip 0.05in} r@{\hskip 0.05in} r@{\hskip 0.05in} r@{\hskip 0.05in} r@{\hskip 0.05in} r@{\hskip 0.05in} r@{\hskip 0.05in} r@{\hskip 0.05in} r@{\hskip 0.05in} r@{\hskip 0.05in} r@{\hskip 0.05in} r@{\hskip 0.05in} r@{\hskip 0.05in} r@{\hskip 0.05in} r@{\hskip 0.05in} r@{\hskip 0.05in} r@{\hskip 0.05in} r@{\hskip 0.05in} }
			\hline
			24&  0& -4&  4&  4&  0&  0&  0&  4& -4&  0&  0& -4& -4&  0& -4&  0&  0&  0&  4&  0&  0&  0&  0\\
			0& 24&  4& -4& -4&  0&  0&  0& -4&  4&  0&  0&  4&  4&  0&  4&  0&  0&  0& -4&  0&  0&  0&  0\\
			-4&  4& 24&  0&  0&  4&  0&  0&  0&  0&  0&  4&  0&  0&  0&  0&  0&  0&  0&  0&  0&  0&  0&  0\\
			4& -4&  0& 24&  0& -4&  0&  0&  0&  0&  0& -4&  0&  0&  0&  0&  0&  0&  0&  0&  0&  0&  0&  0\\
			4& -4&  0&  0& 24& -4&  0&  0&  0&  0&  0& -4&  0&  0&  0&  0&  0&  0&  0&  0&  0&  0&  0&  0\\
			0&  0&  4& -4& -4& 24&  0&  0& -4&  4&  0&  0&  4&  4&  0&  4&  0&  0&  0& -4&  0&  0&  0&  0\\
			0&  0&  0&  0&  0&  0& 24&  0&  0&  0&  0&  0&  0&  0&  0&  0&  0&  0&  0&  0&  0&  0&  0&  0\\
			0&  0&  0&  0&  0&  0&  0& 24&  0&  0&  0&  0&  0&  0&  0&  0&  0&  0&  0&  0&  0&  0&  0&  0\\
\hline
		\end{tabular}
		\caption{$\cD^{\top}\cD$, $\smax=4$, $\freq=36$}
	\end{center}
\end{table}


\begin{table}[ht]
	\begin{center}
		\begin{tabular}{r@{\hskip 0.05in} r@{\hskip 0.05in} r@{\hskip 0.05in} r@{\hskip 0.05in} r@{\hskip 0.05in} r@{\hskip 0.05in} r@{\hskip 0.05in} r@{\hskip 0.05in} r@{\hskip 0.05in} r@{\hskip 0.05in} r@{\hskip 0.05in} r@{\hskip 0.05in} r@{\hskip 0.05in} r@{\hskip 0.05in} r@{\hskip 0.05in} r@{\hskip 0.05in} r@{\hskip 0.05in} r@{\hskip 0.05in} r@{\hskip 0.05in} r@{\hskip 0.05in} r@{\hskip 0.05in} r@{\hskip 0.05in} r@{\hskip 0.05in} r@{\hskip 0.05in} r@{\hskip 0.05in} }
			\hline
			-1&-1& 1&-1& 1& 1&-1&-1& 1& 1&-1& 1& 1&-1& 1&-1&-1&-1&-1& 1&-1&-1&-1&-1& 1\\
			-1&-1&-1&-1& 1& 1& 1& 1& 1& 1& 1& 1&-1& 1&-1& 1& 1& 1& 1& 1& 1&-1&-1& 1& 1\\
			-1& 1& 1& 1& 1&-1& 1&-1&-1&-1&-1& 1& 1& 1& 1& 1&-1& 1& 1& 1&-1& 1& 1& 1& 1\\
			1& 1& 1&-1& 1& 1&-1& 1&-1& 1&-1&-1&-1&-1& 1&-1&-1& 1& 1&-1& 1& 1& 1& 1& 1\\
			-1&-1& 1&-1&-1& 1& 1& 1&-1& 1&-1&-1& 1& 1& 1& 1& 1& 1&-1&-1&-1&-1& 1&-1&-1\\
			1&-1& 1&-1&-1&-1& 1&-1& 1& 1& 1&-1&-1&-1&-1& 1& 1&-1&-1& 1&-1& 1& 1& 1& 1\\
			1&-1&-1& 1&-1& 1&-1&-1& 1&-1& 1& 1& 1&-1& 1& 1& 1& 1& 1&-1& 1& 1& 1&-1& 1\\
			1& 1&-1& 1&-1& 1& 1& 1& 1&-1&-1& 1&-1& 1& 1&-1&-1&-1&-1&-1&-1&-1& 1& 1& 1\\
\hline
		\end{tabular}
		\caption{An E$(s^2)$-optimal, minimax-optimal SSD of size $N=24$, $m=25$, E$(s^2)= 3.84000$}
	\end{center}
\end{table}

\begin{table}[ht]
	\begin{center}
		\begin{tabular}{r@{\hskip 0.05in} r@{\hskip 0.05in} r@{\hskip 0.05in} r@{\hskip 0.05in} r@{\hskip 0.05in} r@{\hskip 0.05in} r@{\hskip 0.05in} r@{\hskip 0.05in} r@{\hskip 0.05in} r@{\hskip 0.05in} r@{\hskip 0.05in} r@{\hskip 0.05in} r@{\hskip 0.05in} r@{\hskip 0.05in} r@{\hskip 0.05in} r@{\hskip 0.05in} r@{\hskip 0.05in} r@{\hskip 0.05in} r@{\hskip 0.05in} r@{\hskip 0.05in} r@{\hskip 0.05in} r@{\hskip 0.05in} r@{\hskip 0.05in} r@{\hskip 0.05in} r@{\hskip 0.05in} }
			\hline
			24& -4& -4&  4&  0&  0&  0&  0&  0&  0&  4&  0&  0&  0& -4& -4&  0&  4&  4& -4& -4&  4&  4&  0&  0\\
			-4& 24&  0&  4&  4&  0&  0&  0&  0& -4& -4&  0& -4& -4&  0&  0& -4&  4& -4& -4&  0&  0&  4&  0&  0\\
			-4&  0& 24&  0& -4&  0&  0&  0&  0&  4&  0&  0&  0&  0&  0&  0& -4&  0&  0&  0&  0&  0&  0&  0&  0\\
			4&  4&  0& 24& -4&  0&  0&  0&  0&  0&  0&  0&  0&  0&  4&  0&  0&  0&  0&  0&  0&  0&  0&  0&  0\\
			0&  4& -4& -4& 24&  0&  0&  0&  4&  0&  4& -4&  0&  0&  4& -4&  0&  4&  4&  0& -4&  4&  0&  0&  0\\
			0&  0&  0&  0&  0& 24&  0&  0&  0&  4&  0&  0&  0&  0&  0&  0& -4&  0&  0&  0&  0&  0&  0&  0&  0\\
			0&  0&  0&  0&  0&  0& 24&  0&  0&  0&  0&  0&  0&  0&  0&  0&  0&  0&  0&  0&  0&  0&  0&  0&  0\\
			0&  0&  0&  0&  0&  0&  0& 24&  0&  0&  0&  0&  0&  0&  0&  0&  0&  0&  0&  0&  0&  0&  0&  0&  0\\
\hline
		\end{tabular}
		\caption{$\cD^{\top}\cD$, $\smax=4$, $\freq=72$}
	\end{center}
\end{table}

\end{document}